%% file: main.tex
\documentclass[12pt]{article}
\usepackage[usenames,dvipsnames]{xcolor}
\usepackage{fullpage,amsmath,amsfonts,amssymb,amsthm,epsfig, graphicx, natbib, appendix,multirow,float,subfig,multicol,hyperref,setspace, tikz}
\newcommand{\bb}{\mathbb}

\definecolor{darkblue}{rgb}{0.0, 0.0, 0.55}
\hypersetup{colorlinks,linkcolor={darkblue},citecolor={darkblue},urlcolor={darkblue}} 
\usepackage[capitalise,noabbrev,nameinlink]{cleveref}

\usepackage{array}
\newcolumntype{P}[1]{>{\raggedright\arraybackslash}p{#1}}

\theoremstyle{definition} \newtheorem{Theorem}{Theorem}
\theoremstyle{definition} \newtheorem{Cor}{Corollary}
\theoremstyle{definition} \newtheorem{Prop}{Proposition}
\theoremstyle{definition} \newtheorem{lemma}{Lemma}
\theoremstyle{definition} \newtheorem{definition}{Definition}
\theoremstyle{definition} 
\theoremstyle{definition} 
\theoremstyle{definition} \newtheorem{remark}{Remark}
\theoremstyle{definition} 
\DeclareMathOperator*{\argmax}{arg\,max}

\usepackage{comment}

\title{Games on Endogenous Networks\thanks{We are grateful to Krishna Dasaratha for a very helpful discussion. We also thank (in random order) Nicole Immorlica,  Navin Kartik, Kevin He, Nageeb Ali, Matt Elliott, Mohammad Akbarpour, Zo\"{e} Hitzig, Antonio Cabrales, Mihai Manea,  Yann Bramoull\'{e}, Giacomo Lanzani, Arun Chandrasekhar, Drew Fudenberg, Emily Breza,  Eran Shmaya, Roberto Corrao, and many seminar and conference participants for helpful conversations. Golub gratefully acknowledges support from the National Science Foundation under grant SES-1629446.  Yann Calv\'{o}-L\'{o}pez, Yu-Chi Hsieh, and Rithvik Rao provided excellent research assistance.}} 
\author{Evan Sadler\thanks{Columbia University, es3668@columbia.edu} \and Benjamin Golub\thanks{Northwestern University, bgolub@northwestern.edu}}
\date{\today}

\begin{document}

\maketitle

\begin{abstract}
\noindent We study network games in which players choose the partners with whom they associate as well as an effort level that creates spillovers for those partners. New stability definitions extend standard solution concepts for each choice in isolation:  pairwise stability in links and Nash equilibrium in actions. Focusing on environments in which all agents agree on the desirability of potential partners, we identify conditions that determine the shapes of stable networks. The first condition concerns whether higher actions create positive or negative spillovers for neighbors. The second concerns whether actions are strategic complements or substitutes to links. Depending on which combination of these conditions occurs, equilibrium networks are either \emph{ordered overlapping cliques} or \emph{nested split graphs}---highly structured forms that facilitate the computation and analysis of outcomes. We apply the framework to examine the consequences of competition for status, to microfound matching models that assume clique formation, and to explain empirical findings in which designing groups to leverage peer effects backfired.
\end{abstract}

\section{Introduction}

\hspace{1 pc}
Social contacts influence people's behavior, and that behavior in turn affects the connections they form.  For instance, a good study partner might lead a student to exert more effort in school, and that student's increased effort may incentivize her partner to maintain the collaboration.  Understanding how networks and actions mutually influence one another is crucial for policy design. 
This paper theoretically analyzes this interplay.

While many models analyze peer effects assuming an exogenously fixed social network, an important paper by \citet*{Carrelletal2013} shows this can lead to mistaken predictions and interventions that backfire.  The authors first estimated academic peer effects among cadets at the U.S. Air Force Academy and subsequently designed peer groups to improve the performance of academically weaker cadets.  Extrapolating from peer-effect estimates based on randomly assigned first-year peer groups, the authors expected that low-skilled freshmen would benefit from being placed in groups with higher proportions of high-skilled freshmen.

When this intervention was actually carried out, it produced a \emph{negative} effect of a size comparable to the positive effect that the model predicted.  The authors interpret this as a consequence of endogenous friendship and collaboration networks \emph{within} administratively assigned groups. Friendships between low- and high-skilled freshmen were much less likely to form in the designed groups than in the randomly assigned ones---effects not considered in the model guiding the intervention. Therefore, the positive spillovers that the design sought to maximize failed to happen.  To account for such effects, researchers need models that permit a simultaneous analysis of network formation and peer effects.

We build a framework to study games with network spillovers together with strategic link formation.  Our theoretical contribution is twofold. First, we propose a model that nests standard analyses of each type of interaction on its own. Players choose actions (in our main application, effort levels) from an ordered set, and they also have preferences over links given actions. We define stability concepts that extend Nash equilibrium (for actions) and pairwise stability (for network formation).   Intuitively, in a solution to a network game with network formation, players should have an incentive to change neither their actions\footnote{We use this term from now on to mean the strategic action \emph{other} than the link choice.}  nor their links.  %

Second, we apply this framework to make predictions about network structure and behavior, focusing on a class of environments in which heterogeneity across players is ``vertical,'' with agreement about who desirable partners are.  Our main results identify ordinal payoff properties, described below, under which stable networks are highly structured. One structure,  \emph{ordered overlapping cliques}, has players grouped into cliques that exhibit homophily, with those taking similar efforts grouped together. This pattern seems important for applications, such as the Air Force study, but has not been theoretically understood. The other type of structure emerging in our results is a \emph{nested split graph}, a strongly hierarchical type of network with a core of highly central players connecting the rest. Our characterization unifies and generalizes some prior work on the emergence of this type of structure. After presenting the main theoretical results, we use the tractable structure they deliver to conduct several analyses on questions of applied interest, including our motivating example. 

In more detail, we first focus on a class of \emph{separable} environments in which the benefit from forming a link depends only on the actions of the two players involved.  We obtain sharp characterizations of outcomes using two kinds of order conditions. The first concerns the nature of spillovers. We say a game has \emph{positive spillovers} if players taking higher actions are more attractive neighbors.  Correspondingly, a game has \emph{negative spillovers} if players taking higher actions are less attractive neighbors.\footnote{Note these spillover properties are distinct from strategic complements or substitutes \emph{in actions}. Our properties concern levels of utility rather than how others' actions affect incentives to take higher actions.}  The second type of condition concerns the relationship between a player's own action incentives and number of links.  The game exhibits \emph{action--link  complements} if the returns from taking higher actions increase with one's degree (number of neighbors) in the network. The game exhibits \emph{action--link substitutes} if these returns decrease with one's degree.  Our main result characterizes the structure of both actions and links for any combination of order conditions (one of each type). \cref{tab:2by2} summarizes our findings.

{\footnotesize

\renewcommand{\arraystretch}{1.5}
\begin{table}[t]
\begin{tabular}{p{2.6cm}  p{3cm} P{4.2cm} P{4.2cm}  }

 & & \multicolumn{2}{c}{\footnotesize \emph{Higher actions create\ldots}} \\
 & & \textbf{positive spillovers} & \textbf{negative spillovers} \\
  \hline
\multirow{2}{2.7cm}{\footnotesize{\emph{Links and actions are strategic\ldots }}} & \textbf{complements}   & Nested split graph, higher degree implies higher action    &  Ordered overlapping cliques, neighbors take similar actions \\
 &  \textbf{substitutes} &    Ordered overlapping cliques, neighbors take similar actions  &  Nested split graph, higher degree implies lower action  \\
 \hline
\end{tabular}
\caption{Summary of main results; each cell indicates which network and action configurations are stable under the corresponding pair of assumptions on payoffs.}\label{tab:2by2}
\end{table}

\renewcommand{\arraystretch}{1}
}

\begin{figure}
    \centering 
    \subfloat[Nested split graph]{{\includegraphics[width=0.3\textwidth]{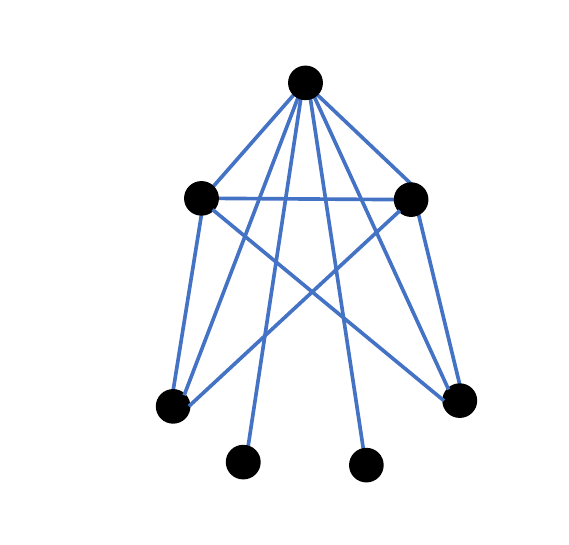} }}\quad \subfloat[Overlapping cliques]{{\includegraphics[width=0.3\textwidth]{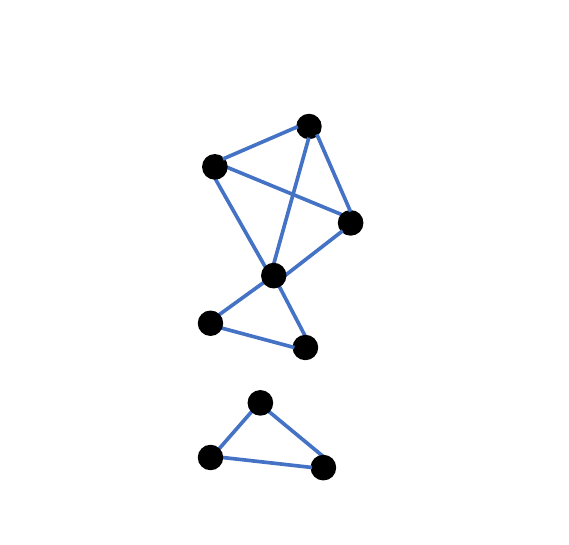} }}\caption{Examples of the two types of networks that are characterized by our main result. Players' actions (effort levels) are ordered according to their vertical position in the diagram.}\label{fig:graphs} 
\end{figure}

\cref{fig:graphs} illustrates the two types of graphs our model predicts.  Nested split graphs are strongly hierarchical networks: nodes are partitioned into classes according to their degrees, and when one node is in a higher class than another, its neighborhood is a strict superset of the other node's neighborhood. On the other hand, ordered overlapping cliques entail that we can order nodes such that every neighborhood is an interval under this ordering, and the endpoints of a node's neighborhood are increasing in the node's own position in the order---imagine a ranking of nodes, and each is connected to all others that are close enough in rank.\footnote{For recent econometric work concerning the estimation of related models, see \citet*{chandrasekhar2021network}.} In each case, equilibrium action levels are ordered in a way that corresponds to the network structure.  To distill the key forces driving these characterizations, we identify ordinal properties of linking incentives that are sufficient for these structures to emerge. These conditions, which we term \emph{consistency} and \emph{alignment}, yield this dichotomy in a much larger class of environments.

We then specialize our framework to interpret the findings of \citet{Carrelletal2013}.  We represent their environment as one of positive spillovers combined with action--link substitutes. That is, we assume that students who spend more time studying create benefits for their peers, but more time studying also makes link formation and maintenance more costly.  We identify conditions for a study group to be completely connected, and different conditions under which it is fragmented.  The message of these characterizations is that the complete graph becomes harder to sustain as types get more spread out---types in the middle are necessary to induce types at the extremes to link with one another.  Small changes in group composition that replace medium-ability cadets with high-ability ones keep the network connected and benefit low-ability cadets through better peer effects. However,  if we make the group more bimodal in ability levels---as in the designed squadrons of \cite{Carrelletal2013}---it endogenously fragments.
When the group divides into two cliques, one with high-ability cadets and one with low-ability students, low-ability cadets no longer experience peer effects from high-ability students, so the benefits of group design disappear. We discuss how the mechanism in our model closely tracks the interpretation that \citet{Carrelletal2013} conjecture for their results based on analysis of post-treatment data.

To further demonstrate the framework, we briefly study two additional applications.  First, we introduce a model of ``status games'' based on \citet*{Immorlicaetal2017}.  Competitions for status entail a combination of action--link complements and negative spillovers. Imagine, for example, people who compete for social status through conspicuous consumption and also choose their network connections. Those who engage in more conspicuous consumption have an incentive to form more social connections to reap status benefits (action--link complements).  On the other hand, those who do so are less attractive friends since linking with them creates negative comparisons (negative spillovers).  In this setting, our model predicts that individuals will sort into cliques with those  who invest similar amounts in status signaling---a finding consistent with stylized facts from sociological studies. Moreover, those in larger friend groups---popular individuals---engage in more conspicuous consumption due to heightened competition, and an increase in status concerns causes the social graph to fragment into smaller cliques.

Our final application provides a microfoundation for  ``club'' or ``group matching'' models. Theories of endogenous matching for public goods or team production often \emph{assume} that spillovers occur in disjoint cliques, which is critical for tractability. The question is then:  which cliques form? We show that even if agents can arrange their interactions into more complex structures if they wish,
there are natural conditions under which cliques are still the predicted outcome. 

Following our applications, we address several important questions that are tangential to the main analysis.  First, we report conditions under which pairwise stable outcomes are guaranteed to exist---while existence is assured in all of our examples, it is not immediate in general because the presence of a link is a discrete outcome, and our solution concept considers joint deviations.  Second, we examine key modeling choices in more detail.  Finally, as our predicted structures are more stark that what we observe in real networks, we discuss two ways to accommodate more complex structures.  %

Endogenizing network structure and behavior together is both theoretically important and challenging.  \citet*{kranton2001theory} studied \emph{market} exchange following strategic network formation, while \citet*{jackson2002formation} introduced an interplay between network formation and a network game (in their case a coordination game). Other important early papers studying  network games and network formation include \citet*{GaleottiGoyal2010}, and \citet*{Cabralesetal2011}.   Our methodological contributions include (i) stability concepts that extend canonical static solution concepts for network formation and equilibrium behavior in isolation; and (ii) an approach to characterizing outcomes via ordinal incentive properties, in the spirit of monotone comparative statics. We also contribute a new substantive prediction: we identify natural conditions under which cliques form among those similar on a ``vertical'' attribute such as productivity.  This type of homophilous structure  is quite different from those studied in the aforementioned theoretical literature, yet it seems important in practice---as highlighted by our applications inspired by \citet{Carrelletal2013},  \citet{Immorlicaetal2017}, and  other work. In contrast, theoretical treatments of the endogenous emergence of homophily have mainly focused on \emph{horizontal} preferences,  such as people preferring to be friends with those who are similar to them (see, e.g., \citet{holme2006nonequilibrium,polanski2023homophily,kivinen2017polarization}). 

Nested split graphs have received considerable attention in earlier network formation models with payoff functions more similar to ours \citep*{Hiller2017,Hellmann2020}---they are thought, for instance, to capture key features of inter-bank lending networks and other trading networks \citep*{Konigetal2014}.  By identifying ordinal payoff properties that produce these structures, we help clarify and unify key forces driving earlier results. For instance, we find that strategic complements in actions are \emph{not} relevant to whether nested split graphs appear---what matters instead is how actions relate to one's attractiveness as a neighbor and one's desire for links.  More detailed connections are drawn to these and other papers in \cref{sec:related}, once we can make reference to the specifics of our results.

\section{Framework} \label{sec:framework}

\hspace{1 pc}
A \textbf{network game with network formation} is a tuple $\langle N, (S_i)_{i \in N}, (u_i)_{i \in N}\rangle$ consisting of the following data:
\begin{itemize}
    \item There is a finite set $N$ of players; we write $\mathcal{G}$ for the set of all simple, undirected graphs on $N$.\footnote{The set $N$ is fixed, so we identify a graph with its set $E$ of \emph{edges} or \emph{links}---an edge is an unordered pair of players.  We write $ij$ for the edge $\{i,j\}$.}
    
    \item For each player $i \in N$, there is a set $S_i$ of actions; we write $\mathcal{S} = \prod_{i \in N} S_i$ for the set of all action profiles.
    
    \item  For each player $i \in N$, there is a payoff function $u_i \, : \, \mathcal{G} \times \mathcal{S} \to \bb{R}$. This gives player $i$'s payoff as a function of a graph $G \in \mathcal{G}$ and a profile of players' actions.
\end{itemize}
A pair $(G,\mathbf{s}) \in \mathcal{G} \times \mathcal{S}$ is an \textbf{outcome} of the game.  Given a graph $G$, we write $G_i$ for the neighbors of player $i$; we write $G +  E$ for the graph $G$ with the links $E$ added and $G-E$ for the graph $G$ with the links $E$ removed.

We next introduce solution concepts that distinguish equilibrium outcomes of this game. We then present conditions on the game---first more intuitive and then more general ones---which are central to our subsequent characterizations.

\subsection{Solution concepts} \label{sec:solution_concepts}

\hspace{1 pc}
Intuitively, in a solution to a network game with network formation, players should have an incentive to change neither their actions nor their links.  We propose two nested solution concepts reflecting this idea. These correspond to existing concepts in the network formation literature, and we extend them to our setting with action choices.\footnote{See \citet*{bloch2006definitions} and \citet*[Chapter 6]{jackson2008}.}  

\begin{definition}
An outcome $(G,\mathbf{s})$ is \textbf{strictly pairwise stable} if the following conditions hold.
\begin{itemize}
    \item The action profile $\mathbf{s}$ is a Nash equilibrium of the game $\langle N,(S_i)_{i \in N}, (u_i(G, \cdot))_{i \in N}\rangle$ in which $G$ is fixed and players only choose actions $s_i$.
    
    \item There is no link $ij \in G$ such that $u_i(G -ij, \mathbf{s}) \geq u_i(G, \mathbf{s})$.
    
    \item There is no link $ij \notin G$ such that both $u_i(G + ij, \mathbf{s}) \geq u_i(G, \mathbf{s})$ and $u_j(G + ij, \mathbf{s}) \geq u_j(G, \mathbf{s})$.
\end{itemize}

\noindent The outcome $(G,\mathbf{s})$ is \textbf{strictly pairwise Nash stable} if additionally there is no pair $(s'_i,H_i)$, consisting of an action $s'_i \in S_i$ and a subset of $i$'s neighbors, $H_i \subseteq G_i$, such that $u_i\left(G - \{ik \, : \, k \in H_i\}, (s'_i, s_{-i})\right) \geq u_i(G, \mathbf{s})$.

\end{definition}

\begin{remark}
Note that our primary definitions require strict preferences over links.  This is technically convenient because it facilitates the cleanest characterization of stable networks, but one could obtain weaker results using analogous solution concepts that permit indifference.\footnote{We say the outcome is \textbf{weakly pairwise stable} if in the second bullet we replace the weak inequality with a strict inequality (there can exist a link $ij$ such that $i$ is indifferent about deleting it), and in the third bullet we require one of the two inequalities to be strict (two players may both be indifferent about adding the missing link between them). For the definition of weak pairwise Nash stability, the weak inequality in the last line becomes a strict inequality---we only require that no player can strictly benefit from deleting a subset of her links.}  Whenever we simply write ``pairwise (Nash) stable,'' this should be understood as ``strictly pairwise (Nash) stable.''  More substantively, both of these solution concepts reflect that link formation requires mutual consent.  An outcome is strictly pairwise stable if $\mathbf{s}$ is a Nash equilibrium given the graph, no player wants to unilaterally delete a link, and no pair of players jointly wish to form a link.  Pairwise Nash stability adds the stronger requirement that no player benefits from simultaneously changing her action and deleting some subset of her links.
\end{remark}

\begin{remark}
Whenever a player considers changing her action, she treats the network as fixed, and whenever two players consider adding a link between them, they treat the action profile $\mathbf{s}$ as fixed. Other modeling choices (e.g., allowing more deviations) may come to mind; \cref{sec:foundations} offers a detailed discussion the methodological reasons for our choice.  However, even if we evaluated a larger set of deviations, results that apply to every pairwise stable outcome still apply to whatever remains. Thus, we see the set of deviations as \emph{the minimal ones needed to ensure the predictions in our results hold}, rather than as a constraint on which deviations players are allowed to consider. 
\end{remark}

\begin{remark}
Note that standard models of network games and strategic network formation are nested in our framework.  To represent a network game on a fixed graph $G$, one can take the utility functions from the network game and add terms so it is strictly optimal for all players to include exactly the links in $G$.  Pairwise stable outcomes in the corresponding network game with network formation correspond to Nash equilibria in the original network game.  To represent a model of network formation, simply make each $S_i$ a singleton (and let the payoffs correspond to those in the network formation model).
\end{remark}
\subsection{Separable network games} \label{sec:separable}

\hspace{1 pc}
We now introduce some structure on payoffs; the environment we study nests several canonical models of network games and permits a simple statement of sufficient conditions for our characterization of network structures.  Suppose all players have the same action set $S_i = S \subseteq \bb{R}$, and payoffs take the form
\begin{equation}\label{eq:parametricpayoff}
u_{i}(G, \mathbf{s}) = v_i(\mathbf{s}) + \sum_{j \in G_i} g(s_i, s_j)
\end{equation}
for each player $i$.  Here, the function $v_i : \mathbb{R}^n \to \mathbb{R}$ is idiosyncratic to player $i$ and captures strategic incentives that are independent of the network.  These payoffs embed two substantive assumptions about the incremental value of any link to its participants, namely that this value:  (i)  is invariant to permutations of players' labels (linking incentives are \emph{anonymous}); and (ii) depends only on the actions of the two players involved (linking incentives are \emph{separable}).  These properties imply that the direct effects of links on agents' payoffs are separable,
which is a canonical feature of network games.\footnote{This refers only to the part of the payoff that affects linking incentives, since $v_i(\mathbf{s})$ is completely arbitrary. Importantly, there may be \emph{indirect} effects of other players' actions on a player's payoffs, but these are mediated through the action choices of neighbors.}  Implicitly, this also means that the cost of forming a link is constant holding actions fixed. For this reason, it is best to view this as a model of network formation within relatively small communities in which one could realistically maintain links with most of the other players.\footnote{In Section \ref{sec:more_complex_payoffs} we discuss how our analysis can accommodate relaxations of both (i) and (ii).}  For formalizations of conditions (i) and (ii) and a statement that these imply the payoff form \eqref{eq:parametricpayoff}, see \cref{sec:formalization_payoff}.

The specification of the payoffs is inspired by canonical network game models in which each individual takes a single action with all counterparties, with some further structure as described above. Indeed, using payoffs of the form \eqref{eq:parametricpayoff}, we can endogenize linking decisions in many common models of network games.  For instance, if we take $v_i(\mathbf{s}) = b_i s_i - \frac{1}{2}s_i^2$, and $g(s_i, s_j) = \alpha s_i s_j$, we obtain a linear-quadratic model similar to that in \citet*{Ballesteretal2006}.  While we can accommodate idiosyncratic private incentives for effort, note that the anonymity assumption precludes varying the weights attached to different links.  However, the bilateral term $g$ can readily incorporate linking costs that depend arbitrarily on player $i$'s action $s_i$, allowing us to capture a wide range of linking incentives.  Despite this flexibility, the payoff form \eqref{eq:parametricpayoff} is still considerably more restrictive than what our results require, and thus the next subsection introduces weaker, purely ordinal assumptions.\footnote{See also \Cref{sec:more_complex_payoffs} on models with more heterogeneity.
}

Within the class of separable network games, we can state simple order conditions that determine the structure of stable networks, and these cover all of our applications. In a network game with network formation, we can view higher actions and additional links as two different activities in which a player can invest; the following definition captures a notion of complementarity between these activities.  

\begin{definition} Given payoffs of the form \eqref{eq:parametricpayoff}, the game exhibits \textbf{action--link complements} if $g$ satisfies a single-crossing condition in its first argument:
$$g(s,r) >  (\geq) \; 0 \quad \implies \quad g(s',r) > (\geq) \; 0$$
whenever $s' > s$.  Actions and links are complements if taking higher actions makes forming additional links more attractive.  The game exhibits \textbf{action--link substitutes} if the above implication holds whenever $s' < s$.\footnote{This definition readily extends to arbitrary network games with network formation---a game exhibits action--link complements if $u_i(G +ij, \mathbf{s})  - u_i(G - ij, \mathbf{s}) \geq (>) \; 0$ implies $ u_i(G+ij, s'_i, s_{-i}) - u_i(G - ij, s'_i, s_{-i}) \geq (>) \; 0)$ whenever $s'_i > s_i$ and action--link substitutes if the implication holds for $s'_i < s_i$.}   Actions and links are substitutes if taking higher actions makes forming additional links less attractive. \end{definition}

Whereas the previous definition captures how increases in the player's action affect that player's willingness to link,  the  next one concerns how such increases affect the player's attractiveness  as a neighbor. 
\begin{definition} The game exhibits \textbf{positive spillovers} if $g$ satisfies a single-crossing condition in its second argument:
$$g(s,r) > (\geq) \; 0 \quad \implies \quad g(s,r') > (\geq) \; 0$$
whenever $r' > r$, and the game exhibits \textbf{negative spillovers} if the above implication holds whenever $r' < r$.\footnote{Again, these definitions naturally extend to arbitrary games---a game exhibits positive spillovers if $u_i(G+ij, \mathbf{s}) - u_i(G - ij, \mathbf{s}) \geq (>) \; 0$ implies $u_i(G + ij, s'_j, s_{-j}) - u_i(G - ij, s'_j, s_{-j}) \geq (>) \; 0$ whenever $s'_j > s_j$ and negative spillovers if the implication holds for $s'_j < s_j$.} \end{definition}

If the game exhibits positive spillovers, players who take higher actions, all else equal, are more attractive neighbors.  This assumption naturally captures situations in which the action $s_i$ represents effort that benefits neighbors (e.g., studying or gathering information).  If the game exhibits negative spillovers, then players who take higher actions are less attractive neighbors.

\begin{remark}
We emphasize that whether a network game with network formation exhibits action--link complements/substitutes, or positive/negative spillovers is unrelated to whether there are strategic complements or substitutes in \emph{actions}:  the actions $s_i$ and $s_j$ of two different players $i$ and $j$ could be (strategic) complements, substitutes, or neither.  This should be clear as the term $v_i(\mathbf{s})$ in \eqref{eq:parametricpayoff}, which is independent of the graph, can depend on the entire profile of actions $\mathbf{s}$ in a completely arbitrary way. Moreover, there is no complementarity or substitutability assumed within the function $g$.  Action--link complements/substitutes is about how taking higher actions affects a given player $i$'s \emph{incentive to form links}. Similarly, positive/negative spillovers is about how other players' actions affect player $i$'s linking incentives.  Our characterization of stable network structures accordingly says nothing about strategic complements or substitutes in actions, and indeed adding such structure places no further restrictions on the networks that can occur in stable outcomes.
\end{remark}

\subsection{More general ordinal conditions} \label{sec:ordinal_assumptions}

\hspace{1 pc}
The single-crossing properties introduced in the previous subsection will turn out to yield orderings on players that capture their preferences for linking and their desirability as partners. These orderings are ultimately the keys to characterizing stable outcomes. Separable payoffs serve as a simple case that motivates these more general properties. This subsection substantiates the claim by introducing order conditions that are sufficient to characterize equilibria.  %
The main economic content of these conditions is that the environment features mainly vertical heterogeneity, both in terms of which partners players desire and how much they want to link.\footnote{In Section \ref{sec:discussion} we comment on the limitations and extensions of this structure.} \Cref{sec:class} will characterize equilibrium actions and behavior in separable network games as a corollary.

In the following definitions, we write
$$\Delta_{ij} u_i(G, \mathbf{s}) = u_i(G + ij, \mathbf{s}) - u_i(G - ij, \mathbf{s})$$
for the marginal value of link $ij$ to player $i$ at outcome $(G,\mathbf{s})$. We write
\begin{align*} 
W_i^+(G,\mathbf{s}) &= \{j \in N \, : \, \Delta_{ji} u_j(G,\mathbf{s}) \geq 0\} \\
S_i^{+}(G,\mathbf{s}) &= \{j \in N \, : \, \Delta_{ji} u_j(G,\mathbf{s}) > 0\} 
\end{align*}
for the set of players with a weak (resp., strict) incentive to link with $i$. We write $i \in W_j^-(G,\mathbf{s})$ if $j \in W_i^+(G,\mathbf{s})$, and analogously define $S_j^-(G,\mathbf{s})$.

\begin{definition}\label{def:consistent}
Given a network game with network formation, linking incentives are \textbf{consistent} if there is no outcome $(G,\mathbf{s})$, and a corresponding collection of players $i,j,k,\ell$, such that \begin{enumerate}
\item $S_i^+(G,\mathbf{s})$ contains $k$ but not $\ell$;
\item $S_j^+(G,\mathbf{s})$ contains $\ell$ but not $k$. \end{enumerate}
\end{definition}

In words, linking incentives are consistent if, given any outcome, the players agree on who is a more desirable linking partner.  There are no players $i,j,k,\ell$ such that $k$ wishes to link with $i$ but not $j$, while $\ell$ wishes to link with $j$ but not $i$.  Note this does not preclude heterogeneity in preferences over links---far from it---but if some player $i$ wishes to link with $k$ but not $\ell$, then any player who wishes to link with $\ell$ necessarily also wishes to link with $k$. 

Fixing a given outcome, we define $i \succeq_{\text{in}} j$ if $S_i^+(G,\mathbf{s}) \supseteq S_j^+(G,\mathbf{s})$. It is easy to see (and we will state formally below) that consistency implies that at every outcome, $\succeq_{\text{in}}$ is a complete weak order, capturing players' desirability as partners.

\begin{definition} \label{def:aligned}
Linking incentives are \textbf{aligned} if the following holds at any outcome $(G,\mathbf{s})$. Consider any three players $i,j,k$ such that, at that outcome, $k \succeq_{\text{in}} j \succeq_{\text{in}} i$. Then (omitting arguments $(G,\mathbf{s})$) \begin{align*} W^-_i \cap W^-_k  \; &\subseteq \; W^-_j \subseteq W^-_i \cup W^-_k,  \\ 
\text{and } \quad S^-_i \cap S^-_k  \; &\subseteq \; S^-_j. \end{align*} \end{definition}

Alignment relates a player's desirability as a linking partner to her own linking incentives.  Suppose $k$ is a more desirable neighbor than $j$, and $j$ is a more desirable neighbor than $i$, so that $j$'s desirability is between that of $i$ and $k$. Alignment says that $j$'s desire for links is also between that of $i$ and $k$. That is, if both $k$ and $i$ wish to link with some other player $\ell$, then $j$ also wishes to link with $\ell$ (the first containment in the definition), and if neither $k$ nor $i$ wish to link with $\ell$, then $j$ also does not wish to link with $\ell$ (the second containment). The second line states that $j$ strictly wishes to link to $\ell$ whenever both $i$ and $k$ do.

Together, the two properties of consistency and alignment allow us to measure players' desirability as partners, and their inclination to form links, on the same one-dimensional scale.  The following lemma states that if consistency and alignment are satisfied, then we can define two orderings $\succeq_{\text{in}}$ and $\succeq_{\text{out}}$. The first ranks agents by how much others want to link to them. The second ranks agents by how much they want to link to others.

\begin{lemma}\label{lem:order}
Suppose a network game with network formation has consistent and aligned linking incentives.  Then, at any outcome $(G,\mathbf{s})$, there exist weak orders $\succeq_{\text{in}}$ and $\succeq_{\text{out}}$ on the players such that (again omitting arguments $(G,\mathbf{s})$) \begin{enumerate} \item $W^+_k \supseteq W^+_j$ and $S^+_k \supseteq S^+_j$ whenever $k \succeq_{\text{out}} j$;
\item $W^-_k \supseteq W^-_j$ and $S^-_k \supseteq S^-_j$  whenever $k \succeq_{\text{out}} j$.
\end{enumerate}
That is, if $k \succeq_{\text{in}} j$, then every player who (strictly) wants to link with $j$ (strictly) wants to link with $k$ as well, and if $k \succeq_{\text{out}} j$, then $k$ (strictly) wants to link with every player with whom $j$ (strictly) wants to link.  

Moreover, the two orders are either identical or directly opposed.
\end{lemma}

\begin{proof}
See Appendix.
\end{proof}

Consistency and alignment entail considerable structure, so it is notable that they are implied by (and, as we will see, strictly weaker than) the natural single-crossing conditions introduced in the last subsection. If a game with payoffs of the form \eqref{eq:parametricpayoff} exhibits action--link complements or substitutes, and positive or negative spillovers, then linking incentives are necessarily consistent and aligned.  Moreover, the orders $\succeq_{\text{in}}$ and $\succeq_{\text{out}}$, expressing players' desirability as neighbors and desire for neighbors, follow the order of the action set $S$. This result does not require $S \subseteq \bb{R}$---the same conclusion holds for any linearly ordered action set.

\begin{lemma}\label{lem:consistentregular}
Suppose a network game with network formation has payoffs of the form \eqref{eq:parametricpayoff}.  If the game exhibits action--link complements or substitutes, and positive or negative spillovers, then linking incentives are consistent and aligned at any outcome.  Moreover:
\begin{enumerate}
    \item If there are action--link complements and positive spillovers, then $i \succeq_{\text{in}} j$ and $i \succeq_{\text{out}} j$ in any outcome with $s_i \geq s_j$.
    
    \item If there are action--link complements and negative spillovers, then $i \preceq_{\text{in}} j$ and $i \succeq_{\text{out}} j$ in any outcome with $s_i \geq s_j$.
    
    \item If there are action--link substitutes and positive spillovers, then $i \succeq_{\text{in}} j$ and $i \preceq_{\text{out}} j$ in any outcome with $s_i \geq s_j$.
    
    \item If there are action--link substitutes and negative spillovers, then $i \preceq_{\text{in}} j$ and $i \preceq_{\text{out}} j$ in any outcome with $s_i \geq s_j$.
    
\end{enumerate}

\end{lemma}

\begin{proof}
By definition, if the game exhibits action--link complements and positive spillovers, then any player $i$ desires more links when $s_i$ is higher and is a more attractive neighbor when $s_i$ is higher.  We immediately see from \eqref{eq:parametricpayoff} that $i \succeq_{\text{in}} j$ and $i \succeq_{\text{out}} j$ whenever $s_i \geq s_j$, and the result follows.  The other three cases are analogous.
\end{proof}

One might reasonably wonder if there are natural \emph{non-separable} payoffs that yield consistent and aligned linking incentives.  As one example, we can endogenize linking decisions in a model of local public goods similar to \citet*{BramoulleKranton2007}.  Suppose actions take values in the positive reals, and payoffs take the form
$$u_i(\mathbf{s},G) = b\left(s_i + \sum_{j \in G_i} \left(s_j - c_i(s_i)\right)\right) - k_i s_i$$
for an increasing concave function $b$.  A neighbor's contribution $s_j$ is a substitute for a player's own contribution $s_i$, and links come at a cost that may depend on a player's own action.  Players who contribute more are clearly more attractive neighbors, so linking incentives are consistent.  Moreover, as long as the cost function $c_i$ is monotone, either increasing or decreasing, linking incentives are also aligned.

\section{The structure of stable graphs}

\hspace{1 pc}
How do properties of the payoff functions $(u_i)_{i \in N}$ affect stable network structures?  This section derives our main results on the taxonomy of stable graphs.  To state the results, we first define two classes of graphs---recall the illustration in \cref{fig:graphs}.
\begin{definition}
A graph $G$ is a \textbf{nested split graph} if $d_j \geq d_i$ implies that $G_i \subseteq G_j \cup \{j\}$.

\vspace{1 pc}

\noindent A graph $G$ consists of \textbf{ordered overlapping cliques} if we can order the players $\{1,2,\ldots,n\}$ such that $G_i \cup \{i\}$ is an interval $I_i \subseteq \{1,2,\ldots,n\}$, for each $i$, and the endpoints of this interval $I_i$ are weakly increasing in $i$.

\end{definition}

In a nested split graph, neighborhoods are ordered through set inclusion, resulting in a strong hierarchical structure.\footnote{See \cite{Konigetal2014} and \cite*{belhaj2016efficient} for more on the structure and other economic properties of these networks.  A common alternative characterization is the following.  Given a graph $G$, let $\mathcal{D} = (D_0, D_1,\ldots,D_k)$ denote its degree partition---players are grouped according to their degrees, and those in the (possibly empty) element $D_0$ have degree $0$.  The graph $G$ with degree partition $\mathcal{D}$ is a nested split graph if and only if for each $\ell$ and each $i\in D_{\ell}$, we have
\[G_i = \left[\,\bigcup_{j=1}^\ell D_{k + 1 - j}\,\right] \cap N_{-i}, \] in which $N_{-i}$ denotes players other than $i$---taking the intersection with this set is necessary because for $\ell>k/2$, the union includes $i$'s own partition element.} In a graph with ordered overlapping cliques, the player order induces an order on the set of maximal cliques.  Each maximal clique consists of an interval of players, and both endpoints of these cliques are strictly increasing.  Any graph in which every component is a clique is a special case of this structure.  

\subsection{General characterization}

\hspace{1 pc}
Our first theorem shows that, if linking incentives are consistent and aligned, then stable graphs have one of the two structures we have defined.  This finding relies on the two orders derived in \cref{lem:consistentregular}.  If players who are more attractive also desire more links, we are naturally led to nested split graphs with attractive players in more central positions.  If players who are more attractive desire fewer links, these opposing incentives produce ordered overlapping cliques.

\begin{Theorem}\label{theo:structure}
Suppose a network game with network formation has consistent and aligned linking incentives, and $(G,\mathbf{s})$ is a strictly pairwise stable outcome.  Then:
\begin{enumerate}
    \item If the orders $\succeq_{\text{in}}$ and $\succeq_{\text{out}}$ are identical, the graph $G$ is a nested split graph in which players higher in the two orders have higher degrees.
    
    \item If the orders $\succeq_{\text{in}}$ and $\succeq_{\text{out}}$ are opposed, the graph $G$ consists of ordered overlapping cliques with respect to either order.
    
\end{enumerate}

\end{Theorem}

\begin{proof}
We begin with part (a).  Fixing the outcome $(G, \mathbf{s})$, suppose $j \succeq_{\text{in}} i$ and $j \succeq_{\text{out}} i$.  This means that every $k$ that wants to link with $i$ also wants to link with $j$, and $j$ wants to link with every $k$ with whom $i$ wants to link.  Since $(G, \mathbf{s})$ is strictly pairwise stable, there can be no indifference about links, so any $k \neq j$ that is a neighbor of $i$ must be a neighbor of $j$.

For part (b), we show that if $i \succeq_{\text{in}} j \succeq_{\text{in}} k$ and $ik \in G$, then also $ij \in G$ and $jk \in G$---note this implies $i \preceq_{\text{out}} j \preceq_{\text{out}} k$.  Since $ik \in G$, we know $i$ wants to link with $k$, and since $j \succeq_{\text{in}} k$, this means $i$ wants to link with $j$.  Since $i$ wants to link with $j$, and $k \succeq_{\text{out}} i$, we know $k$ wants to link with $j$.  Similarly, since $i$ wants to link with $k$ and $j \succeq_{\text{out}} i$, we know $j$ wants to link with $k$.  Since $i \succeq_{\text{in}} k$, this implies $j$ wants to link with $i$.  Since $(G, \mathbf{s})$ is strictly pairwise stable, there can be no indifference about links, so we conclude that $ij \in G$ and $jk \in G$.
\end{proof}

The characterization in \cref{theo:structure} is stark.  There are essentially two network structures that can arise in stable outcomes:  either neighborhoods are nested, or the network is organized into overlapping cliques of players.  In case (a), if one player is ranked higher than another in the two orders, then the two neighborhoods are ordered by set inclusion.  In case (b), a link between two players implies that the set of players ranked in between the two forms a clique.  Strict comparisons play an important role as any link $ij$ need not be in $G$ if both $i$ and $j$ are indifferent about adding it.

\begin{remark} \label{rem:exogenous_actions}
Implicit in this result is a novel characterization of structures that arise in pure network formation games---the definitions of consistent and aligned link preferences do not change if action sets are singletons.  In this case, it is as if each player has a one-dimensional type, and linking incentives depend on these types---higher types are more (or less) attractive neighbors, and a higher-ranked player either always has a stronger incentive to form links (or always has a weaker incentive to form links).  Work on strategic network formation has thus far faced challenges in obtaining general results on the structure of pairwise stable graphs, and \cref{theo:structure} highlights non-trivial conditions that yield sharp predictions.
\end{remark}

\subsection{Separable network games}\label{sec:class}

\hspace{1 pc}
Recall the separable games of \cref{sec:separable}. \cref{lem:consistentregular} showed that single-crossing conditions in separable games imply consistent and aligned linking incentives. Combining this with \cref{theo:structure}, we obtain the following characterization of stable network structures in separable games.

\begin{Cor}\label{cor:structure}
Suppose a network game with network formation has payoff functions of the form \eqref{eq:parametricpayoff}.  If $(G,\mathbf{s})$ is a strictly pairwise stable outcome, then:
\begin{enumerate}
    \item If the game exhibits action--link complements and positive spillovers, then $G$ is a nested split graph in which players with higher degrees take higher actions.
    
    \item If the game exhibits action--link substitutes and negative spillovers, then $G$ is a nested split graph in which players with higher degrees take lower actions.
    
    \item If the game exhibits action--link complements and negative spillovers, or action--link substitutes and positive spillovers, then $G$ consists of ordered overlapping cliques with respect to the order of players' actions.
\end{enumerate}
\end{Cor}

\begin{proof}
The result is immediate from \cref{theo:structure} and \cref{lem:consistentregular}.
\end{proof}

\begin{remark}
While payoffs of the form \eqref{eq:parametricpayoff} encompass all of our applications in the next section, we want to highlight that \cref{theo:structure} applies much more broadly.  There are at least two natural extensions of this class of games for which an analogous result is immediate.  First, one could replace the function $g(s_i, s_j)$ in the sum with a term of the form $g(s_i, s_j) - c_i(s_i)$---having added an idiosyncratic cost of linking, the conclusions of \cref{cor:structure} continue to hold for $g(s_i, s_j)$ increasing/decreasing in each argument.  Second, one could replace $g(s_i, s_j)$ with a term of the form $g\left(h_i(s_i), h_j(s_j)\right)$, in which $\{h_i\}_{i \in N}$ are arbitrary idiosyncratic functions of the players' actions. The orders $\succeq_{\text{in}}$ and $\succeq_{\text{out}}$ would then follow the order of the function values $\{h_i\}_{i \in N}$ rather than the order of players' actions.
\end{remark}

\begin{remark}
Studies of network formation often evaluate the efficiency of pairwise stable graphs, but we should note that general results are not possible here.  In the case of separable network games, the idiosyncratic term $v_i(\mathbf{s})$ could entirely dominate the match specific terms when thinking about efficiency.  Even if we limit attention to the match specific terms, whether actions exert positive or negative externalities is orthogonal to whether actions are strategic complements or substitutes.  The efficiency of stable outcomes ultimately depends on other features of the payoff functions, which will vary across settings.
\end{remark}

\section{Perverse consequences of group design}\label{sec:carrell}

\hspace{1 pc}
We now turn to our first application, using the tractability our characterizations afford to incorporate a network formation analysis into the peer effects model of \citet{Carrelletal2013}.  \citet{Carrelletal2013} estimated academic peer effects among first-year cadets at the US Air Force Academy (using a standard model that took networks as exogenous) and then used these estimates to inform the assignment of new cadets to squadrons---administratively designed peer groups of about 30 cadets.  Based on a first cohort of randomly assigned squadrons, the authors concluded that being in a squadron with higher-performing peers\footnote{Specifically, those entering with relatively high scores on the verbal section of the SAT exam.} led to better academic performance among less prepared cadets.  In the treatment group of a later cohort, incoming cadets with less preparation were systematically placed in squadrons with larger numbers of high-ability peers.  While the researchers' goal was to improve the performance of the less prepared cadets,\footnote{The precise objective they were maximizing was the performance of the bottom third of cadets.} the intervention ultimately backfired:  these students performed significantly worse.  In this section, we present a model showing that our theory can simultaneously explain two distinctive features of the Air Force study:
\begin{enumerate}
    \item When peer group composition changes slightly, low-ability cadets are better off when they have more high-ability peers, and
    
    \item Larger changes in peer group composition eliminate or even reverse this effect.
\end{enumerate}

\noindent Broadly, our results show that stable graphs become more fragmented if private returns to effort $b_i$, which can be interpreted as ability levels, are more heterogeneous.  Thus, placing cadets of high and low abilities together, without cadets of middle ability to bridge the gap, can result in isolated cliques that eliminate the desired spillovers.

Methodologically, this section illustrates how our framework permits a tractable analysis of  link formation and action choice suited to one of our motivating applications. Importantly, we will see that the endogeneity of actions is crucial to explain the stylized facts. The ``bridging'' role of the middle-ability cadets makes sense only because of their ability to influence the actions of the low-ability cadets to whom they are linked.

\subsection{Payoffs} We now present a tractable specification of our model suited to the application. Consider a network game with network formation in which $S_i = \bb{R}_+$ for each player $i$, and payoffs take the form
$$u_i(G, \mathbf{s}) = b_i s_i + \alpha s_i \sum_{j \in G_i} s_j - \frac{1}{2}(1 + d_i) s_i^2,$$
in which $d_i = |G_i|$ is player $i$'s degree, and $\alpha \in [0,1]$.  Holding the graph fixed, this is a standard linear-quadratic network game of strategic complements.  Taking $v_i(\mathbf{s}) = b_i s_i - \frac{1}{2}s_i^2$ and $g(s_i, s_j) = \alpha s_i s_j - \frac{1}{2}s_i^2$, we see it also falls into the class \eqref{eq:parametricpayoff}.  There are positive spillovers, as an increase in $s_j$ makes a link to player $j$ more valuable.  Moreover, links and actions are substitutes as $g(s_i, s_j)$ satisfies the requisite single crossing property---as $s_i$ increases, the benefit to $i$ of linking to $j$ decreases and eventually turns negative, implying those who invest a lot of effort find linking too costly.\footnote{A natural interpretation is that studying and socializing each take time away from the other activity.  While studying together can also strengthen social ties, the substantive assumption here is that a marginal hour studying together is less conducive to friendship formation than that same hour spent together on a leisure activity.}$^{, }$\footnote{One might alternatively use a payoff function with a hard resource constraint split between studying and socializing---our structural results would still apply---but we believe a flexible allocation is more realistic.  Cadets spend time on other activities, such as sleep and solitary leisure, that can also be reallocated.}  One can readily check that in a pairwise stable outcome, players $i$ and $j$ are neighbors only if $\frac{s_j}{2 \alpha} \leq s_i \leq 2 \alpha s_j$.

The first-order conditions for actions imply that, in any pairwise-stable outcome,
$$s_i = \frac{1}{1 + d_i}\left(b_i + \alpha \sum_{j \in G_i} s_j\right)$$
for each $i \in N$.  Writing $\tilde{G}$ for a matrix with entries $\tilde{g}_{ij} = \frac{1}{d_i + 1}$ if $ij \in G$ and $0$ otherwise, and $\tilde{\mathbf{b}}$ for a column vector with entries $\frac{b_i}{d_i + 1}$, we can express this in matrix notation as
$$\mathbf{s} = \tilde{\mathbf{b}} + \alpha \tilde{G} \mathbf{s} \quad \implies \quad \mathbf{s} = (I - \alpha \tilde{G})^{-1}\tilde{\mathbf{b}}.$$
For $\alpha \in [0,1]$, the solution for $\mathbf{s}$ is unique and well-defined in any graph $G$, and it is an equilibrium of the game holding $G$ fixed. We can also compute payoffs:  notice that when $i$ plays her best response $s_i$, her payoff in the graph $G$ is
$$u_i = \frac{1}{2}(1 + d_i) s_i^2.$$
Hence, if $(G, \mathbf{s})$ and $(G', \mathbf{s}')$ are two pairwise stable outcomes, player $i$ is better off under $(G', \mathbf{s}')$ if and only if $(1+d'_i)(s'_i)^2 > (1 + d_i) s_i^2$.

\subsection{The structure of stable outcomes}

\hspace{1 pc}
How do private incentives $b_i$, and the strength of spillovers $\alpha$, affect the set of stable outcomes? Here we derive sufficient conditions for the network to be completely fragmented. We will also derive conditions that are sufficient, and others that are necessary, for a complete graph to be part of a stable outcome. These results show that,  as a general rule, stable graphs become more fragmented if spillovers $\alpha$ are small, if private incentives $b_i$ are more spaced out, and if the population size $n$ is large.  In the next subsection, we further specialize the model to  apply these insights to the \citet{Carrelletal2013} setting.

Our first proposition characterizes conditions under which the empty graph is part of a pairwise stable outcome.  Note that if $G$ is empty, the unique equilibrium actions are $s_i = b_i$ for each player $i$.  In the following results, we always assume players are ordered so that $b_1 \leq b_2 \leq \cdots \leq b_n$.

\begin{Prop}\label{prop:empty}
The empty graph is part of a pairwise stable outcome if and only if private incentives are sufficiently spaced out: $\frac{b_{i+1}}{b_i} \geq 2\alpha$ for each $i = 1,2,\ldots,n-1$.  Moreover, there exists a threshold $\underline{\alpha} \geq \frac{1}{2}$ such that, whenever $\alpha < \underline{\alpha}$, no nonempty graph is possible in a pairwise stable outcome.  If the $b_i$ are all distinct, then the threshold satisfies $\underline{\alpha} > \frac{1}{2}$.  Otherwise, we have $\underline{\alpha} = \frac{1}{2}$.

\end{Prop}

\begin{proof}
See Appendix.
\end{proof}

The first part of \cref{prop:empty} tells us that the empty graph is stable whenever the private incentives are sufficiently spaced out.  How spaced out they need to be is increasing in the strength of spillovers.  Moreover, if the spillover parameter $\alpha$ is small enough, then the empty graph is the only graph that can appear in a pairwise stable outcome.  Our next result provides an analogous characterization of conditions under which the complete graph is part of a pairwise stable outcome.

 \begin{Prop}\label{prop:complete} $\;$

\begin{enumerate} 
    \item There exists a pairwise stable outcome $(G, \mathbf{s})$ in which $G$ is complete if
\begin{equation}\label{eq:completeexist}
\frac{b_n}{b_1} \leq \frac{\alpha(1 + n)}{(1-\alpha)(2\alpha + n)}.
\end{equation}
\item Moreover, if 
\begin{equation}\label{eq:completeunique}\frac{b_n}{b_1} < \frac{2\alpha}{\alpha + (1-\alpha)(n-1)},
\end{equation}
then there is a unique pairwise stable outcome, and in it, $G$ is complete.  

\item Conversely, if
$$b_n\left(2 \alpha^2 + n(1 - 2\alpha^2)\right) > b_1 \alpha \left(4 \alpha - 1 + 2n(1-\alpha)\right),$$
then the complete graph is not part of any pairwise stable outcome.
\end{enumerate}
\end{Prop}

\begin{proof}
See Appendix.\end{proof}

The main message of \cref{prop:complete} is that the complete graph becomes harder to sustain as private incentives $b_i$ get more spread out:  if the ratio $\frac{b_n}{b_1}$ is too high, there are pairwise stable outcomes with disconnected graphs, and the complete graph may not be part of any stable outcome.  An important consequence of the first claim is that the complete graph is part of a pairwise stable outcome whenever $\alpha$ is sufficiently close to $1$.  If $\alpha = 1$, the complete graph is always part of a pairwise stable outcome. More generally, stronger spillovers encourage more connected graphs. The proof implies that claims (a) and (b) are tight whenever $b_1 = b_2 = \cdots = b_{n-1}$---if \eqref{eq:completeunique} fails, there is a pairwise stable outcome in which the first $n-1$ players form a clique, and player $n$ is isolated.  As $n$ gets larger, or $\alpha$ gets smaller, the inequalities \eqref{eq:completeexist} and \eqref{eq:completeunique} become harder to satisfy, and outcomes with disconnected graphs become more likely.

The proof works by first establishing \Cref{lem:cliqueaction} in \Cref{sec:carrell_support}, which characterizes the payoff to players in a clique in terms of  underlying parameters---the $b_i$, $\alpha$, and the size of the clique. This allows us to analyze the conditions for clique stability by examining the incentives of the players most tempted to deviate.

Note that even though the results we have stated are specific to complete graphs, the analysis readily generalizes to subsets of players, giving both necessary and sufficient conditions for cliques to form among them.

\subsection{The importance of intermediate-ability types: An illustration}

\hspace{1 pc}
Specializing the model to include just three productivity types allows us to transparently relate the model to the findings of \citet{Carrelletal2013}, illustrating how the  groups they designed can fail to benefit low-ability cadets because of the absence of middle types to knit the group together.  

Let $\alpha = 1$, and private incentives take one of three values $b_i \in \{b_\ell, b_m, b_h\}$, satisfying the following inequalities: $2 < b_h / b_\ell < 4$, $b_h / b_m < 2$, and $b_m / b_\ell < 2$.  We can interpret type $b_\ell$ as having low ability, type $b_h$ as having high ability, and type $b_m$ as having intermediate ability.  Given an outcome $(G, \mathbf{s})$, we interpret the action $s_i$ as the academic performance of cadet $i$, and links are friendships through which peer effects operate.

Since $\alpha = 1$, the complete graph is always part of a pairwise stable outcome.  However, it should also be clear that without any middle types, an outcome with two isolated cliques---one consisting of low types taking action $s = b_\ell$ and one consisting of high types taking action $s = b_h$---is also pairwise stable.  This outcome is clearly worse for low-type players as they take both lower actions and have fewer connections.  Moreover, this outcome is the only one that survives a natural refinement.  We call a pairwise stable outcome $(G,\mathbf{s})$ \textbf{uncoordinated} if there exists a sequence of graphs and action profiles $(G^{(0)}, \mathbf{s}^{(0)}, G^{(1)}, \mathbf{s}^{(1)},\ldots)$, ending at $(G, \mathbf{s})$, in which
\begin{itemize}
    \item $G^{(0)}$ is empty,
    \item $\mathbf{s}^{(k)}$ is a Nash equilibrium holding $G^{(k)}$ fixed, and
    \item we have $G^{(k+1)} = G^{(k)} + ij$ for some pair $ij$, and both $u_i(G^{k+1}, \mathbf{s}^{(k)}) \geq u_i(G^{k}, \mathbf{s}^{(k)})$ and $u_j(G^{k+1}, \mathbf{s}^{(k)}) \geq u_j(G^{k}, \mathbf{s}^{(k)})$ with at least one strict inequality.
\end{itemize}

\noindent In words, a pairwise stable outcome is uncoordinated if it is reachable through myopically beneficial link additions, starting from an empty graph and assuming that players reach a Nash equilibrium action profile following each new link.\footnote{This selection criterion implicitly assumes that players have no prior relationships at the start of the adjustment process.  Given information on prior relationships, one could adapt this criterion to select an outcome reachable from the initial state.}  

When both the complete graph and segregated cliques are pairwise stable, there is good reason to expect the latter outcome in practice---uncoordinated stable outcomes formalize this idea.  However, if we include enough middle types, we can always eliminate the bad pairwise stable outcome.

\paragraph{Three illustrative squadrons} Suppose there are $5$ cadets, and we take $\{b_\ell, b_m, b_h\} = \{4,6,9\}$.  We now assess stable outcomes for three different squadron compositions:
\begin{itemize}
    \item Squadron 1:  $\mathbf{b} = (4,4,6,6,9)$
    
    \item Squadron 2:  $\mathbf{b} = (4,4,6,9,9)$
    
    \item Squadron 3:  $\mathbf{b} = (4,4,9,9,9)$
\end{itemize}

\begin{figure}
    \centering 
    \subfloat[Squadron 1]{
    \resizebox{0.27\textwidth}{!}{    \input{illustrations/squadron1.tex}}
    }\quad 
    \subfloat[Squadron 2]{
    \resizebox{0.27\textwidth}{!}{    \input{illustrations/squadron2.tex}}
    }\quad 
    \subfloat[Squadron 3]{
    \resizebox{0.27\textwidth}{!}{    \input{illustrations/squadron3.tex}}
    }
    \caption{An illustration of the stable outcomes for the three squadrons. Ability levels $b_i$ appear inside each node, while equilibrium actions $s_i$ are next to the node.}
    \label{fig:squadrons} 
\end{figure}
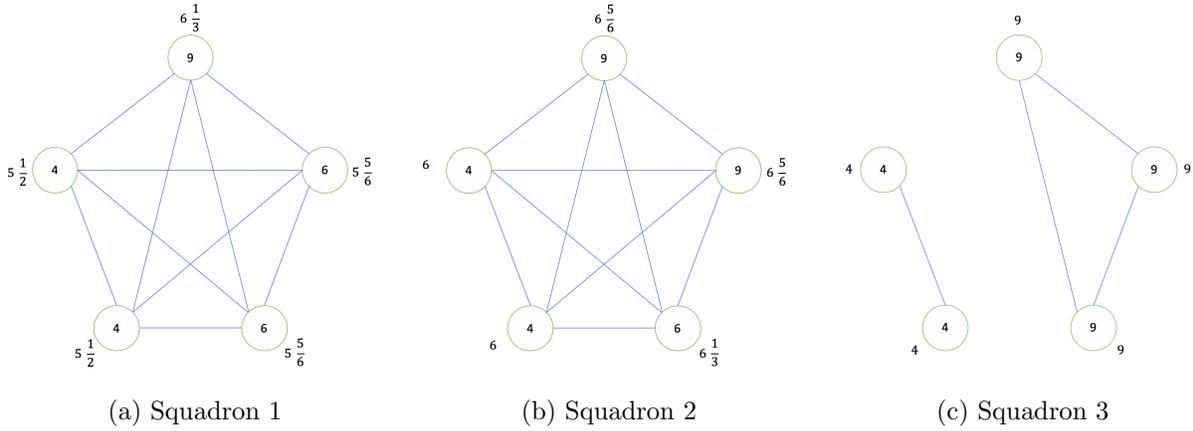

\noindent In each successive squadron, we replace a cadet of intermediate ability with one of high ability, and we are interested in how the actions and welfare of the low-ability cadets change. In the context of \cite{Carrelletal2013}, the first two squadrons represent combinations that should occur frequently in the chance assignments of the first cohort, while the last one represents the designed groups.\footnote{This is consistent with the facts reported in \cite{Carrelletal2013}, who note that the protocol for designing squadrons, in order to group low-ability cadets with  many high-ability peers, tended to exclude cadets of intermediate ability and place these into more homogeneous squadrons---we do not discuss these here.}

In the first two squadrons, the unique pairwise stable outcome involves a complete graph.  In squadron 1, the action vector is $\mathbf{s} = \left(5 \frac{1}{2}, 5 \frac{1}{2}, 5\frac{5}{6},5\frac{5}{6},6 \frac{1}{3}\right)$, and in squadron 2, the action vector is $\mathbf{s} = \left(6,6,6\frac{1}{3}, 6 \frac{5}{6}, 6 \frac{5}{6}\right)$.  From this we see that adding a second high-ability cadet to the squadron increases the performance of low-ability cadets from $5 \frac{1}{2}$ to $6$, and low-ability cadets benefit from this small change in group composition.  What happens if we add another high-ability cadet?  In squadron 3, the unique uncoordinated pairwise stable outcome involves two separate cliques:  the two low-ability cadets form one clique, the three high-ability cadets form the other, and the action vector is $\mathbf{s} = (4,4,9,9,9)$.  A larger change in the group composition results in a marked decline in performance for the low-ability cadets.

\subsection{Discussion}

\hspace{1 pc}
The predicted outcome in the designed squadrons is fragmentation: a critical mass of high-ability cadets forms their own clique, which is consistent with the explanation that  \citet{Carrelletal2013} give for the unintended consequences they observed. Surveying their participants subsequently, they found that treatment squadrons had significantly higher rates of ability homophily than control squadrons.\footnote{These calculations controlled for opportunities to link.} Indeed, even though the designed squadrons had more high-ability cadets (compared to a random squadron), the low-ability cadets in those squadrons were actually less likely to have high-ability cadets as study partners. This is strong evidence of a new force preventing the peer effects from materializing. Our theory explains this via the strategic forces arising from endogenous network formation.

It is worth noting that endogenous actions were a key part of this story. Imagine  a model in which agents' actions are replaced by exogenous ability levels (as in \Cref{rem:exogenous_actions}), and externalities remain separable across relationships. In this case, adding intermediate-ability cadets cannot change whether the high-ability cadets want to link to the low-ability ones.

\section{Other applications}

\hspace{1 pc}
We briefly address two additional applications.  We first study \emph{status games}, in which payoffs incorporate social comparisons---this allows us to interpret stylized facts about clique formation.  We then discuss how our analysis can provide a foundation for group matching models that assume a coarse  notion of network formation in which agents choose cliques in which to participate. Prior work studying network games together with network formation focuses predominantly on the case of action--link complements and positive spillovers, and such models predict nested split graphs.  In contrast, our main applications fall into cells in which stable graphs consist of ordered overlapping cliques.

\subsection{Status games and ordered cliques} \label{sec:statusgames}

\hspace{1 pc}
A natural model of competitions for status features action--link complements and negative spillovers.  For instance, when people care about their relative position in their social neighborhood, (i) having more friends who engage in conspicuous consumption creates stronger incentives to consume more, but at the same time (ii) those who consume conspicuously are less attractive as friends.  \cite*{jackson2019friendship} argues that many social behaviors (e.g., binge drinking) have the same properties:  those with more friends find these behaviors more rewarding, but they exert negative externalities across neighbors (e.g., due to health consequences or crowding out more productive behaviors).  More generally, this pattern applies to any domain in which friends' achievement drives one to excel, but there is disutility from negative comparisons among friends.  Our theory entails that such situations drive the formation of social cliques ordered according to both their popularity and their effort at the activity in question.

This prediction agrees with anthropological and sociological studies documenting the pervasiveness of ranked cliques. For instance, \citet*{davis1967structure} formalize the theory of \citet*{homans1950human}, asserting that small or medium-sized groups (e.g., departments in workplaces, grades in a school) are often organized into cliques with a ranking among them in terms of their sociability and status-conferring behaviors.\footnote{\citet{davis1967structure}  discuss purely graph-theoretic principles that guarantee some features of a ranked-cliques graph, but do not have a model of choices.} \citet*{adler1995dynamics} conduct an ethnographic study of older elementary-school children that highlights the prevalence of cliques.  The authors argue that status differentiation is clear across cliques, and indeed that there are unambiguous orderings, with one clique occupying the ``upper status rung of a grade'' and ``identified by members and nonmembers alike as the `popular clique." This study also emphasizes the salience of status comparisons with more popular individuals, consistent with our negative spillovers assumption. Building on this ethnographic work, \cite*{gest2007features} carry out a detailed quantitative examination of the social structures in a middle school, with a particular focus on gender differences.  The authors' summary confirms the ethnographic narrative: ``girls and boys were similar in their tendency to form same-sex peer groups that were distinct, tightly knit, and characterized by status hierarchies.''

Within the economics literature, \citet{Immorlicaetal2017} introduce a framework in which players exert inefficient effort in a status-seeking activity and earn disutility from network neighbors who exert higher effort---we can view this as a model of conspicuous consumption with upward-looking comparisons.  The authors assume an exogenous network and explore how the network structure influences individual behavior.  Formally, the authors take $S_i = \bb{R}_+$ for each player $i$, and payoffs are
$$u_i(\mathbf{s}) = b_i s_i - \frac{s_i^2}{2} - \sum_{j \in G_i} g_{ij}\max\{s_j - s_i, 0\},$$
in which $g_{ij} \geq 0$ for each $ij \in G$.  The paper shows that an equilibrium partitions the players into classes making the same level of effort, and the highest class consists of the subset of players that maximizes a measure of group cohesion.  Our framework makes it possible to endogenize the network in this model. Under a natural extension of the payoff function, the classes that emerge in equilibrium form distinct cliques in the social graph.

Consider a network game with network formation in which $S_i = \bb{R}_+$ for each $i$, and payoffs take the form
$$u_i(G, \mathbf{s}) = b s_i - \frac{s_i^2}{2} + \sum_{j \in G_i} \left(1 - \delta \max\{s_j - s_i, 0\} \right).$$
In this game, player $i$ earns a unit of utility for each neighbor,\footnote{Note the graph here is unweighted.} but suffers a loss $s_j - s_i$ if neighbor $j$ invests more effort. There are no other linking costs. To highlight the role of network formation, rather than individual incentives, we also specialize the model so that all players have the same private benefit $b$ for effort.  The game clearly falls into the class \eqref{eq:parametricpayoff} of separable network games, with negative spillovers and weak links-action complements.  Hence, stable outcomes consist of ordered overlapping cliques, and we can only have $ij \in G$ if $|s_i - s_j| \leq \frac{1}{\delta}$.  For the purposes of this example, we restrict attention to outcomes in which the cliques partition the players.  Moreover, following \citet{Immorlicaetal2017}, we focus on maximal equilibria of the status game, with players taking the highest actions they can sustain given the graph.  Since all players have the same private benefit $b$, \emph{all players in a clique play the same action, and the maximum equilibrium action in a clique of size $k$ is $b + (k-1) \delta$.}

This characterization implies two notable features of stable outcomes.  First, those in large groups take higher actions---popular individuals invest more in status signaling.  Second, as status concerns increase, the graph can fragment.  Let $c^*$ denote the smallest integer such that $c^* \delta \geq \frac{1}{\delta}$---this is the unique integer satisfying $\delta \in [1/\sqrt{c^*}, 1/\sqrt{c^*-1})$.  If $i$ and $j$ are in different cliques, we must have $|s_i - s_j| \geq \frac{1}{\delta}$, which implies the cliques differ in size by at least $c^*$.  The larger $c^*$ is, the more cohesive stable networks are.  If there are $n$ players in total, and $\delta < \frac{1}{\sqrt{n-2}}$, then the complete graph is the only stable outcome.  As $\delta$ increases, meaning there are greater status concerns, then stable outcomes can involve more fragmented graphs.  If $\delta \geq 1$, then separate cliques need only differ in size by one player, and the maximal number of cliques is the largest integer $k$ such that $\frac{k(k+1)}{2} \leq n$ (which is approximately $\sqrt{2n}$).

\subsection{Foundations for group-matching models} \label{sec:groupmatching}

\hspace{1 pc}
Models of endogenous matching that go beyond standard pair matching frameworks often posit that individuals belong to a \emph{group} of others. Externalities and strategic interactions then occur within or across groups---with the crucial feature that payoffs are invariant to permutations of agents within groups. In essence, these models constrain the network that can form, assuming disjoint cliques. For example, \citet*{baccara2013homophily} study a setting in which individuals join groups (e.g., social clubs) and then choose how much to contribute to an activity within the group. These contributions affect the payoffs of other group members symmetrically. Similarly, \citet*{chade2018matching} model the allocation of experts to teams.  These experts share information within their teams, benefiting all team members equally, but not across teams. %

The interactions motivating these models are not so constrained in reality---there is no reason why pairs cannot meet outside the groups, and in many cases a person could choose to join multiple groups.  However, assuming that interactions happen in groups allows simplifications that are essential to the tractability of these models.  To what extent are these restrictions without loss of generality?  Our results allow us to provide simple sufficient conditions.

For this section, we assume the common action set $S$ is a closed interval in $\bb{R}$, and each player has one of finitely many types---write $t_i \in T$ for player $i$'s type.  Payoffs take the form
\begin{equation}\label{eq:groupmatch}
u_i(G, \mathbf{s}) = v(s_i, t_i) + \sum_{j \in G_i} g(s_i, s_j),
\end{equation}
in which $v$ and $g$ are continuous.  We further assume that players have unique best responses, holding the graph and other players' actions fixed.  Write $s^*_t = \argmax_{s \in S} v(s,t)$ for the action that a type $t$ player would take if isolated with no neighbors---this is the \emph{privately optimal action}.  Payoffs exhibit a \emph{weak preference for conformity} if player $i$'s optimal action always lies somewhere in between her privately optimal benchmark and the actions of her neighbors.  That is, for $\hat{s} = \argmax_{s_i \in S} u_i(G, s_i, s_{-i})$, we have
$$\min\{s^*_{t_i}, \min_{j \in G_i}\{ s_j \} \} \leq \hat{s} \leq \max\{s^*_{t_i}, \max_{j \in G_i}\{ s_j \} \}$$
for all $i$ and $G$.

We say that types form \emph{natural cliques} if there exists a partition $\{T_1, T_2,\ldots,T_K\}$ of $T$ such that
\begin{itemize}
    \item $g\left(s^*_t, s^*_{t'}\right) \geq 0$ for any $t,t' \in T_k$ and any $k$.
    \item Either $g\left(s^*_t, s^*_{t'}\right) \leq 0$ or $g\left(s^*_{t'}, s^*_t\right) \leq 0$ with at least one strict inequality for any $t \in T_k$ and $t' \in T_\ell$ with $k \neq \ell$.
    
\end{itemize}

\noindent In words, this means that if all players were to choose their privately optimal actions, and form the network taking those actions as given, then disjoint cliques based on the partition of types would be pairwise stable.  If payoffs exhibit a weak preference for conformity, these same cliques remain pairwise stable when players can change their actions.

\begin{Prop}\label{prop:disjoint}
Suppose a network game with network formation has payoffs of the form \eqref{eq:groupmatch}, exhibits a weak preference for conformity, and types form natural cliques.  If the game exhibits either positive spillovers and action--link substitutes or negative spillovers and action--link complements, then there exists a pairwise stable outcome in which the network is exactly the partition into natural cliques.
\end{Prop}

\begin{proof}
We carry out the proof assuming positive spillovers and action--link substitutes---the other case is analogous.  Since types form natural cliques, there is a partition $\{T_1,T_2,\ldots,T_K\}$ of types such that, when playing the privately optimal actions, players have an incentive to link if and only if their types are in the same element of the partition.  Suppose this graph forms.  We show it is part of a pairwise stable outcome.

For each $T_k$ let $\underline{s}_k$ and $\overline{s}_k$ denote the lowest and highest values respectively of $s^*_t$ for some type $t \in T_k$.  Continuity together with weak preference for conformity implies that there exists an equilibrium in actions with $s_i \in [\underline{s}_k, \overline{s}_k]$ for every player $i$ with type $t_i \in T_k$.  Given two such players $i$ and $j$, we have
$$g(s_i, s_j) \geq g(s_i, \underline{s}_k) \geq g(\overline{s}_k, \underline{s}_k) \geq 0,$$
in which the first inequality follows from positive spillovers, and the second follows from action--link substitutes.  Hence, these two players have an incentive to link.

For two partition elements $T_k$ and $T_\ell$, with $k \neq \ell$, assume without loss of generality that $\underline{s}_\ell \geq \overline{s}_k$.  For player $i$ with type $t_i \in T_k$ and $j$ with type $t_j \in T_\ell$ we have
$$g(s_i, s_j) \leq g(\underline{s}_k, s_j) \leq g(\underline{s}_k, \overline{s}_\ell) < 0,$$
so the players have no incentive to link.
\end{proof}

Under mild assumptions, stable networks preserve natural cleavages between identifiable types of individuals, and players endogenously organize themselves into disjoint cliques as assumed in group matching models.  Even if the natural cleavages are not so stark, our results show that much of the simplifying structure remains:  individuals can be part of multiple groups, but each group is a clique, and there is a clear ordering among the cliques.  Imposing this slightly weaker assumption in models of group matching may allow for richer analysis while preserving the tractability that comes from group matching assumptions.

\section{Existence}\label{sec:exist}
\hspace{1 pc}
While pairwise stable outcomes exist in all of our applications, we have not yet addressed the general question of the existence of pairwise stable outcomes.  There are two reasons why existence is non-trivial in our setting.  First, the presence or absence of a link is a discrete event, so we cannot use standard arguments that rely on continuity.  Second, pairwise stability requires the absence of profitable joint deviations to form new links.  Nevertheless, there are natural sufficient conditions that ensure existence of pairwise stable outcomes.  In what follows, we assume that players' action sets are complete lattices with order $\geq$.

\begin{definition}\label{def:convexcomp}
A network game with network formation exhibits \textbf{strategic complements} if for any graph $G$, any $s'_i > s_i$, and any $s'_{-i} > s_{-i}$, we have
$$u_i(G, s'_i, s_{-i}) \geq (>) \; u_i(G, s_i, s_{-i}) \quad \implies \quad u_i(G, s'_i, s'_{-i}) \geq (>) \; u_i(G, s_i, s'_{-i}).$$
The game exhibits \textbf{convexity in links} if for any profile $\mathbf{s}$, any graph $G$, any pair $ij$, and any collection of edges $E$, we have
$$\Delta_{ij} u_i(G, \mathbf{s}) \geq (>) \; 0 \quad \implies \quad \Delta_{ij} u_i(G+E, \mathbf{s}) \geq (>) \; 0.$$
\end{definition}

A network game with network formation exhibits strategic complements if, holding the graph fixed, the underlying normal form game exhibits strategic complements.  The definition imposes a single-crossing condition on players' strategies, which implies that best responses are weakly increasing in others' actions.  The game exhibits convexity in links if, holding actions fixed, adding links to the network weakly increases players' incentives to form links.  Note that in all of our examples, linking incentives are independent of $G$ holding others' actions fixed, so this condition trivially holds.

To state our result, we also need to extend the notions of action--link complements/substitutes and positive/negative spillovers to arbitrary games.

\begin{definition}\label{def:linkcompsub}
A network game with network formation exhibits \textbf{action--link complements} if 
$$\Delta_{ij} u_i(G, \mathbf{s}) \geq (>)\; 0 \quad \implies \quad \Delta_{ij} u_i(G, s'_i, s_{-i}) \geq (>)\; 0)$$
whenever $s'_i > s_i$.  The game exhibits \textbf{action--link substitutes} if the above inequality holds whenever $s'_i < s_i$.
\end{definition}

\begin{definition}\label{def:spillovers}
A network game with network formation exhibits \textbf{positive spillovers} if 
$$\Delta_{ij} u_i(G, \mathbf{s}) \geq (>)\; 0 \quad \implies \quad \Delta_{ij} u_i(G, s'_j, s_{-j}) \geq (>)\; 0$$
whenever $s'_j > s_j$.  The game exhibits \textbf{negative spillovers} if the above inequality holds whenever $s'_j < s_j$.
\end{definition}

\begin{Prop}\label{prop:exist1}
Suppose a network game with network formation exhibits strategic complements and convexity in links.  If either
\begin{enumerate}
    \item the game exhibits action--link complements and positive spillovers, or
    \item the game exhibits action--link substitutes and negative spillovers,
\end{enumerate}
then there exists a pairwise stable outcome.  Moreover, the set of graphs that occur in pairwise stable outcomes contains a maximal and minimal element.\footnote{In fact, one can also show using convexity in links that the minimal graph is part of a pairwise Nash stable outcome.}
\end{Prop}

\begin{proof}
Since the game exhibits strategic complements, for any fixed $G$ there exist minimal and maximal Nash equilibria of the induced normal form game---this follows from standard arguments using Tarski's fixed point theorem.  Likewise, since the game exhibits convexity in links, for any fixed profile $\mathbf{s}$ there exist minimal and maximal pairwise stable graphs.  To get the minimal graph, start from an empty graph and iteratively add links that pairs of players jointly wish to form.  Convexity in links implies that no player will later wish to remove a link that was added earlier, so we must eventually terminate at a stable graph.  Similarly to get the maximal graph, start from a complete graph and iteratively delete links that one of the players wishes to remove.

We now define two maps $\overline{B}(G,\mathbf{s})$ and $\underline{B}(G, \mathbf{s})$ mapping outcomes to outcomes.  Let $\overline{B}(G, \mathbf{s})$ return an outcome $(\overline{G}, \overline{\mathbf{s}})$ in which $\overline{G}$ is the maximal pairwise stable graph given $\mathbf{s}$, and $\overline{\mathbf{s}}$ is the maximal Nash equilibrium given $G$.  Similarly, let $\underline{B}(G, \mathbf{s})$ return an outcome $(\underline{G}, \underline{\mathbf{s}})$ in which $\underline{G}$ is the minimal pairwise stable graph given $\mathbf{s}$, and $\underline{\mathbf{s}}$ is the minimal Nash equilibrium given $G$.

In case (a), positive spillovers and action--link complements imply that the graphs $\overline{G}$ in $\overline{B}(G,\mathbf{s})$ and $\underline{G}$ in $\underline{B}(G, \mathbf{s})$ are weakly increasing in $\mathbf{s}$---holding the rest of the graph fixed, higher $\mathbf{s}$ makes link $ij$ more desirable to both player $i$ and player $j$.  Similarly, positive spillovers and action--link complements imply that the profiles $\overline{\mathbf{s}}$ in $\overline{B}(G, \mathbf{s})$ and $\underline{\mathbf{s}}$ in $\underline{B}(G,\mathbf{s})$ are weakly increasing in $G$.  This means that both $\overline{B}$ and $\underline{B}$ are monotone maps with respect to the natural product order on $\mathcal{G} \times \mathcal{S}$, so Tarski's theorem implies minimal and maximal fixed points exist for both---the maximal fixed point of $\overline{B}$ is the maximal pairwise stable outcome, and the minimal fixed point of $\underline{B}$ is the minimal pairwise stable outcome.

Case (b) follows from similar reasoning after reversing the order on action profiles.  Negative spillovers and action--link substitutes implies that the graphs in $\overline{B}(G, \mathbf{s})$ and $\underline{B}(G, \mathbf{s})$ are weakly decreasing in $\mathbf{s}$, and the profiles are weakly decreasing in $G$, so we can again apply Tarski's theorem.\end{proof}

\cref{prop:exist1} only applies within two out of four cells for the class of games in \cref{sec:class}.  In general, we cannot ensure existence for the other two cells, as the following example illustrates.  Suppose there are two players with a common action set $S = \{0,1\}$ and payoffs
\[
u_i(G, \mathbf{s}) = \begin{cases} s_i & \text{if } G \text{ is empty} \\ 2s_{-i} + \frac{2 s_i s_{-i} - 1}{4} - s_i & \text{if } G \text{ is complete.} \end{cases}
\]
For player $i$, the marginal value of a link to player $-i$ is
$$2(s_{-i} - s_i) + \frac{2 s_i s_{-i} - 1}{4}.$$
This is increasing in $s_{-i}$, and for the relevant range of values it is decreasing in $s_i$---the game exhibits positive spillovers and action--link substitutes.  Moreover, it should be clear that, if the two players are linked, the game exhibits strategic complements.  Nevertheless, there is no pairwise stable outcome.  In the empty graph, each player optimally takes action $1$, and the marginal value of adding the link between them is $\frac{1}{4} > 0$, so they should form the link.  In the complete graph, each player optimally takes action $0$, and the marginal value of the link is then $-\frac{1}{4} < 0$, so they should each drop the link.

Even though existence in not always assured, the structural result in \cref{cor:structure} greatly simplifies the process of searching for a stable outcome.  As we have already seen in three applications, starting with a clique structure and checking whether it is stable often provides a simple way to establish existence.

We note that our existence result extends the main finding in \citet*{Hellmann2013}, obtained in a setting with network formation only. The paper shows that pairwise stable graphs exist if payoffs are convex in own links, and others' links are complements to own links.  These conditions are jointly equivalent to convexity in links in \cref{def:convexcomp}.

\section{Discussion} \label{sec:discussion}

\subsection{Modeling choices in our solution concept}\label{sec:foundations}

\paragraph{Why consider link and action deviations separately?} Our notion of pairwise (Nash) stability is a static solution concept, with stability being defined by the absence of particular individual and pairwise deviations.  Players consider link deviations holding actions fixed and action deviations holding links fixed. As we have already noted, none of our predictions about structure would change if we allowed more deviations---for instance, if players could simultaneously revise both actions and links.  As long as players are \emph{allowed} to contemplate revising links and actions separately, the proofs implying the structured forms in our main results  go through.

To the extent that the specific choice of solution concept does matter, we believe there are good methodological reasons to adopt our definitions. In particular, this choice offers the most natural extension of Nash equilibrium in actions to games in which players jointly take actions and form links.  Suppose we permitted arbitrary joint deviations by any pair of players in both their links and actions, so that, e.g., two players forming a link could jointly agree to change their actions as well.\footnote{Or perhaps only players with an existing or newly forming link would be allowed to revise actions, capturing the the idea that links allow coalitional behavior, in the spirit of \citet{aumann1988endogenous}.}  If we permitted such deviations, the solution concept would no longer specialize to Nash equilibrium when we fix the graph.  While more cooperative solution concepts are worth exploring, our goal is to provide a simple and portable extension of the large literature on Nash equilibrium in network games, and our choice of solution concept is critical for this connection. Indeed, our solution concept is the ``minimal'' extension of both Nash equilibrium in network games and pairwise stability in network formation.

Moreover, even though pairwise stability for linking is defined in a cooperative manner---permitting joint deviations by two players---one can often obtain the same prediction by appealing to variations of trembling hand perfection in a non-cooperative link announcement game \citep{CalvoArmengolIlkilic2009,IlkilicIkizler2019}.\footnote{The idea is that two players who both want a link that is absent should each offer it, just in case another player trembles and accepts; incentives will be driven by single-link additions, since these are the overwhelmingly likely consequences of trembles.} Applying this approach to our setting would entail an equilibrium notion defined by the deviations we have studied, rather than richer multilateral deviations.   %
Still, one might ask:  Why not allow a broader set of unilateral deviations, allowing players to change actions and drop links at the same time?  This is precisely what pairwise Nash stability does, and whenever pairwise Nash stable outcomes exist, all of our results apply.

\paragraph{Why not have a stage of network formation prior to action choice?}

We could have taken a different modeling approach, in which a network formation stage would occur before players take actions. After network formation in an initial period $t=1$, a network game of effort choice (without any link choices) would occur at time $t=2$. Equilibrium play at $t=2$ would determine payoffs, and one could then apply a concept such as pairwise stability to make predictions about network formation.

This timing entails an important substantive assumption:  Agents have commitment power to maintain, or to refuse, links.  This commitment power clearly has strategic implications, and models with it will yield different predictions.  In practice, such commitment power over links after actions are taken seems absent in relevant settings (e.g., friendship formation). Players constantly have opportunities to revise both links and actions, and our solution concept reflects this.

\paragraph{Richer dynamic models}

As we have noted, our main goal is to extend the useful and versatile existing static models of network games and network formation. An alternative approach would consider the revision of actions and links in some kind of dynamic process. For instance, a player revising a link might plan on opponents changing actions in response to that, etc. 

Static models avoid this. The definition of a Nash equilibrium in a static game does not consider further responses after a potential deviation, and neither does our generalization for network games with network formation. Dynamic models that contemplate farsighted revision sequences often end up being very sensitive to the exact modeling of the timing and other aspects of revision opportunities. This presents theoretical challenges that have proved formidable in related settings, and so a static benchmark model generalizing standard static solution concepts is the natural place to start the analysis.

\subsection{More complex payoffs and network structures}
\label{sec:more_complex_payoffs}
\hspace{1 pc}
Our predictions about the structure of stable networks are stark. Real networks are typically not organized precisely into ordered cliques, nor are neighborhoods perfectly ordered by set inclusion.  Nevertheless, our results provide a benchmark for analyzing the qualitative interplay between equilibrium in actions and stable network formation.  In this section we discuss several natural directions to extend our analysis.  

\paragraph{Vertical and horizontal heterogeneity} As  \Cref{sec:ordinal_assumptions} demonstrated, conditions that make the heterogeneity among players ``vertical'' in nature are central to our results. We showed that in a class of environments, stark orderings of linking incentives are implied by natural single-crossing conditions---for instance, when the action of interest (e.g., greater study effort) increases one's attractiveness to any partner. 

Nevertheless, in groups with considerable horizontal heterogeneity, we would not expect consistency and alignment to hold globally; different subgroups might rank partners quite differently.  We leave such analysis for future work because the ``purely vertical'' setting is rich in its own right and can, on its own, offer new explanations of important phenomena, such as the perverse effects of group design.   \citet{Sadler2023} extends the theory a more general ``vertical'' model that accommodates convexity in linking costs, which is absent from our separable games model.

Handling more heterogeneity is an important frontier. Within the framework of the current paper, we can move in that direction by studying notions of consistency and alignment restricted to subsets of agents. For example, consider a group of people sharing some covariates (e.g., demographics), who may be embedded in a larger network. Consistency and alignment---with the ordering corresponding, e.g., to effort at work---may hold when restricted to agents in this group, even if these properties do not hold globally. For instance, we could relax anonymity in the separable games framework, allowing the payoffs from a neighbor to depend on both parties' demographic attributes; group-restricted  consistency and alignment would then follow from single-crossing properties conditional on group membership. Then our results would immediately characterize the network structure within the group. Thus, we see our contribution as highlighting some stark consequences of vertical orderings in linking incentives, which will also have implications in richer environments. 

\paragraph{Multiplex networks} Another approach to modeling more realistic networks is to layer different relationships on top of one another in a ``multiplex'' network---rigid patterns across different layers can combine to form more realistic arrangements. Consider a simple example with two activities: work on the weekdays---in which the activity is production---and religious services on the weekends---in which the activity is attendance and engagement.  Both entail positive spillovers, but work exhibits action--link substitutes---forming friendships takes time that could be devoted to production---while church exhibits action--link complements---attendance makes it easier to form ties. Assuming there is enough heterogeneity in ability or preferences, a non-trivial network will form through each activity. In the work network, we get ordered cliques. In the church network, we get a nested split graph, with the more committed members serving as a core connecting all others. Layering these networks on top of each other can produce a complex network with aspects of both ``centralization,'' mediated by the weekend ties, and homophily, driven by the work ties. This description ties into Simmel's account, subsequently developed by many scholars, of cross-cutting cleavages.

\paragraph{Stochastic payoffs} A third approach is to introduce noise.  \citet{Konigetal2014} provide an example, describing a dynamic process in which agents either add or delete one link at a time, and the underlying incentives exhibit positive spillovers and action--link complements.  If agents always make the myopically optimal link change, the graph is a nested split graph at every step of the process.  However, if agents sometimes make sub-optimal changes, then all graphs appear with positive probability, but the distribution is still heavily skewed towards those with a nested structure.  This allows the authors to fit the model to real-world data.  Based on our analysis, one could adapt this model to study peer effects or status games, and under suitable assumptions obtain a noisy version of our ordered cliques prediction. More generally, one could define stochastic generalizations of our ordinal properties and study the implications for the distribution over stable network structures.

\section{Related work} \label{sec:related}

\hspace{1 pc}
Our analysis sits at the intersection of two strands of work in network theory: games on fixed networks and strategic network formation.  It is most closely related to research that synthesizes the two.  

\paragraph{Network games.} Within the network games literature, the most widely used and tractable models feature real-valued actions and linear best replies or similarly tractable parameterizations \citep{Ballesteretal2006, BramoulleKranton2007,Bramoulleetal2014}. Our main results derive predictions from order conditions on the payoff functions and the action space, rather than particular functional forms.  Nevertheless, canonical models from the above-mentioned literature provide important leading examples.  See \citet*{Sadler2020a} for ordinal characterizations in a setting with exogenous networks. 

\paragraph{Network formation without action choices.} Among models in which the only endogenous choice is linking, \citet*{Hellmann2013,Hellmann2020} has some of the most closely related results.  Our existence result in \cref{sec:exist} directly extends that in \citet{Hellmann2013}, and \citet*{Hellmann2020} predicts nested split graphs in a class of network formation models that falls within one of our four cases.\footnote{Relatedly, \citet{JacksonWatts2001} provide a general existence result for pairwise stable graphs based on a potential function, and \citet*{ChakrabartiGilles2007} extend this analysis.  One can formulate similar results in our setting.}  The earliest precedent for ordered overlapping cliques comes from \citet*{JohnsonGilles2000}, in which they are called ``locally complete.'' These graphs arise as pairwise stable outcomes, but the mechanism---quite different from ours---is that players in a ``spatial connections model'' have an assumed preference to link with those who are geographically closer in one-dimensional space. In contrast, we assume a common ranking of the most desirable partners. Most importantly, our work differs from this strand in that preferences over links come from equilibrium payoffs of a network game rather than exogenous externalities of links themselves.\footnote{See, e.g., \citet*{JacksonWolinsky1996} and \citet*{GoyalJoshi2006} for canonical analyses with exogenous link payoffs.}

More recently, \citet*{Sadler2023} builds on an earlier version of the present paper, characterizing stable network structures based on ordinal payoff properties similar to those used here.  A key difference is that \citet{Sadler2023} explicitly focuses on payoffs with convex linking costs, which necessitates using a refinement of pairwise stability.  Convex costs may preclude nested split graphs as these networks require some players to maintain a large number of links, so predictions necessarily differ.

\paragraph{Network formation with endogenous actions.} Thus far, the literature on strategic choice of effort (or related actions) in endogenous networks is small.  \cite{jackson2002formation} observed that combining network formation and action choice could change predictions relative to those implied by studying either dimension separately,\footnote{They demonstrated this in a two-action coordination game with linear linking costs, focusing on stochastically stable play.} while \cite{Cabralesetal2011} studied network formation with a single global search investment in a setting interpreted as a large anonymous community---quite different from the small-group settings to which our analysis is suited.  Many subsequent papers made an important simplification in assuming players form links \emph{unilaterally}, in contrast with our model requiring mutual consent. For instance, in \citet{GaleottiGoyal2010}, players unilaterally form links, the proposer of a link incurs the cost, and a public goods game is subsequently played.  Equilibrium networks involve a core-periphery structure, and ex-ante identical players endogenously specialize---some invest in information and others invest in links.  Similarly, the model of \citet*{Baetz2015} entails unilateral link formation together with a linear-quadratic game of strategic complements, leading to strongly hierarchical network structures.  One key difference is that decreasing marginal returns to linking cause those at the top of the hierarchy to refrain from linking with one another.  In \citet*{HerskovicRamos2020}, agents receive exogenous signals and form links to observe others' signals, and they subsequently play a beauty contest game. A player whose signal is observed by many others exerts greater influence on the average action, which in turn makes this signal more valuable to observe.  The equilibrium networks again have a hierarchical structure closely related to nested split graphs. %

\citet*{JoshiMahmud2016}, \citet*{Hiller2017}, and \citet*{Badev2021} are closer to our approach.  \citet{Badev2021} studies a binary action coordination game with endogenous link formation, proposing a solution concept that interpolates between pairwise stability and pairwise Nash stability. This is a parametric model for estimation and simulation procedures in a high-dimensional environment; the goal is to study empirical counterfactuals rather than deriving theoretical results.  \citet{Hiller2017} studies a game in which each player chooses a real-valued effort level and simultaneously proposes a set of links---a link forms if and only if both players propose it.  The author then refines the set of Nash equilibria by ruling out pairwise deviations in which two players create a link between them and simultaneously adjust their actions.  In the underlying game, players have symmetric payoff functions that exhibit strategic complements and positive spillovers.  The setting therefore falls into the first cell in our table, and even though the solution concept differs slightly from ours, the resulting outcomes are indeed nested split graphs.\footnote{Within the pure network formation literature, without action choice, \citet{Hellmann2020} studies a network formation game in which all players are ex-ante identical and uses order properties of the payoff functions to characterize the architecture of stable networks.  A key result shows that if more central players are more attractive linking partners, then stable networks are nested split graphs.  By specifying an appropriate network game, one can view this finding as a special case of the action--link complements and positive spillovers cell in our table.}  \citet{JoshiMahmud2016} take a somewhat different approach, modeling link proposals followed by action choices in a canonical linear-quadratic game.  Their analysis includes both local and global interaction terms, but it still produces nested split graphs in the relevant cells.   Relative to this work, which corresponds to the nested split graph cells of our taxonomy, we significantly relax parametric and symmetry assumptions on players' payoffs, highlighting more fundamental properties that lead to this structure. 

In a related but distinct effort, \citet{Konigetal2014} study a dynamic network formation model in which agents choose strategic actions and myopically add and delete links.  Motivated by observed patterns in inter-bank lending and trade networks, the authors seek to explain the prevalence of hierarchical, nested structures.  The underlying incentives satisfy positive spillovers and action--link complements, and accordingly the stochastically stable outcomes are nested split graphs. Most recently,  \citet*{CerreiaVoglio2023Nonlinear} apply our notion of action-link substitutes in coordination games where there are neither positive nor negative spillovers. Using nonlinear fixed-point theory, they are able to characterize equilibrium links and actions in high-coordination limits. 

Our results on the endogenous emergence of ordered 
 overlapping cliques constitute a new substantive prediction, one that plays a central role in our applications inspired by \citet{Carrelletal2013} and \citet{Immorlicaetal2017}. Similarly motivated by \citet{Carrelletal2013},  \citet*{bolletta2021model} argued that such a structure arises in a model of players forming a pairwise stable network and then playing a coordination game on it, but this turned out to be false \citep*{golub2023difficulty}, and a tractable approach to endogenous fragmentation in such settings has remained an open problem.  Cliques that sort individuals on endogenous outcomes have appeared in quite different settings---typically in models involving opinion formation, a horizontal attribute. In those models, agents by assumption simply do not like to be friends with those whose opinions are too far from their own.\footnote{See, for instance, \citet*{holme2006nonequilibrium}, \citet*{kivinen2017polarization}, and \citet*{polanski2023homophily}.} The analysis is importantly different when, as in our setting, players share a common ranking of partners' desirability. 
 
    \paragraph{Work related to our applications} Several other literatures connect to our applications, as we have mentioned throughout in the context of \citet{Carrelletal2013} and \citet{Immorlicaetal2017}. \cite*{ghiglino2010keeping} is an important antecedent on status games with exogenous networks; see also \cite*{bramoulle2022loss} for a recent contribution.  For two of the cells in \cref{tab:2by2}, our results state that stable structures consist of ordered cliques, and the members of a clique share similar attributes. In some cases, these cliques are disjoint.  One can view this result as providing a microfoundation for group matching models.  In these models, players choose what group to join rather than what links to form, so it is assumed ex ante that the graph is a collection of disjoint cliques.  For instance, \citet{baccara2013homophily} study a model in which players choose to join clubs (i.e., cliques) before investing in club goods, finding that stable clubs exhibit homophily.\footnote{In other related work, \citet*{bandyopadhyay2020pricing} study pricing for group membership in a similar setting, and \citet{chade2018matching} study the allocation of experts to disjoint teams.}  Our analysis extends this finding, and one can use our results to find conditions under which the group matching assumption is without loss of generality.

\section{Final remarks}

\hspace{1 pc}
From academic peer effects to social status to trading networks, the connections people and firms choose to form affect the strategic actions they take and vice versa. Sound behavioral predictions and policy recommendations depend on taking these interactions into account, rather than studying each aspect separately. We offer a flexible framework that unites two types of models and solution concepts, pertaining to strategic actions and link choice. This unification enriches what we can capture and also supports tractable characterizations of equilibrium network structures and behavior.  Several important applications fit within this framework, and we highlight new insights that emerge from applying our results. 

One key point---for which our applications serve as a proof-of-concept---is that the structural predictions the theory offers greatly reduce the space of possible networks, as well as the actions they can support. This is crucial for the tractability of both numerical calculations and theoretical analyses of how equilibria and welfare depend on the environment.

Our framework is best suited to small communities with mainly ``vertical'' heterogeneity. The  conditions we impose rule out, e.g., significant  differences in linking costs due to geography, or convexities in linking costs that are likely to arise in a large group.\footnote{\citet{Sadler2023} undertakes a natural generalization to convex linking costs.} Our taxonomy and applications show that the vertical model delivers some interesting forces on its own. We have also discussed elaborations of the framework that can extend the insights to other settings. First, the characterizations we have derived apply if the conditions on payoffs hold restricted to a suitably defined group of players. For instance, we could restrict attention to a group of students that have the same demographics, with the main remaining heterogeneity being a vertical ability or activity level.  A different extension concerns overlaying different networks corresponding to different kinds of relationships. This raises issues about interactions among relationships studied in the literature on multiplex or multilayer networks. A final promising extension focuses on introducing noise in incentives or link measurements, which connects to the econometrics of networks.

{\footnotesize \singlespacing
\bibliographystyle{ecta}
\bibliography{refs}
}

\newpage
{
\appendix

\section{Omitted proofs and details}

\subsection{Formalization of anonymity and separability conditions} \label{sec:formalization_payoff}

We introduce the key payoff conditions. For this purpose, we fix a numerical representation $u_i(G,\mathbf{s})$ for each player's preferences over outcomes (which are taken to be complete and transitive), which has both ordinal and cardinal content (insofar as it captures the intensity of preferences for various links, or equivalently preferences over lotteries).

\begin{definition} Let $\sigma$ be a permutation of the agents and let $G_\sigma$ and $\mathbf{s}_{\sigma}$ denote the graph and action profile after agents are relabeled according to $\sigma$.\footnote{So that $G_{\sigma}$ has link $\sigma(i)\sigma(j)$ if and only if $G$ has link $ij$, and $s_{\sigma(i)}=s_{i}$.}  Linking incentives are \textbf{anonymous} if incentives to form links do not depend on players' labels: player $i$ strictly prefers to add  link $ij$ at $(G,\mathbf{s})$ if and only if agent $\sigma(i)$ strictly prefers to add link $\sigma(i)\sigma(j)$ at  $(G_\sigma,\mathbf{s}_\sigma)$. \end{definition}

Our next definition, of separability, posits that the incremental values of links are separable.\footnote{This definition could, of course, be restated in terms of choice behavior over lotteries for links.}  

\begin{definition} Linking incentives are \textbf{separable} if the following holds for every $i$, $j$, and $\mathbf{s}$. Player $i$'s valuation of adding the  link $ij$, $$ u_i(G+ij,\mathbf{s})-u_i(G,\mathbf{s}), $$ depends only on $s_i$ and $s_j$, and not on any other actions or links. 
\end{definition}

In this section we establish a simple result: 
\begin{lemma}
Linking incentives are separable and anonymous if and only if payoffs can be represented via the form (\eqref{eq:parametricpayoff}), which we reproduce here for convenience:
\begin{equation}\label{eq:parametricpayoff2}
u_{i}(G, \mathbf{s}) = v_i(\mathbf{s}) + \sum_{j \in G_i} g(s_i, s_j).
\end{equation} \end{lemma}

\begin{proof} The ``if'' direction is trivial. For the ``only if'' direction, we suppose that linking incentives are anonymous and separable. Fix any $\mathbf{s}$. Let $v_i(\mathbf{s})$ be defined as $u_i(\varnothing,\mathbf{s})$, where $\varnothing$ is the empty graph. Now let the neighbors of $i$ in $G$ be enumerated as $j_1, j_2, \ldots,j_d$. Let $G_k$ be the graph with edges $(i,j_{k'})$ for all $k' \leq k$, with $G_0$ understood to be the empty graph $\varnothing$. By separability of linking incentives, we can inductively show $$ u_i(G,\mathbf{s}) =v_i(\mathbf{s})+\sum_{k=1}^d [u_i(G_k,\mathbf{s})-u_i(G_{k-1},\mathbf{s})].$$ Separability further implies that $u_i(G_k,\mathbf{s})-u_i(G_{k-1},\mathbf{s})$ is a function, say $g_k(s_i,s_j)$, only of $s_i$ and $s_j$. Thus we write $$ u_i(G,\mathbf{s}) =v_i(\mathbf{s})+\sum_{k=1}^d g_k(s_i,s_{j_k}).$$
Anonymity then requires that $g_k$ does not depend on $k$. Thus we may write $$ u_i(G,\mathbf{s}) =v_i(\mathbf{s})+\sum_{k=1}^d g(s_i,s_{j_k}),$$ which is the desired form.
\end{proof}

\subsection{Proofs of results in \cref{sec:framework}}

\paragraph{Proof of Lemma 1}
Fixing an outcome $(G, \mathbf{s})$, define the binary relation $\succeq_{\text{in}}$ via $i \succeq_{\text{in}} j$ if there is no player $k$ such that
$$\Delta_{ik} u_k(G,\mathbf{s}) < 0 \quad \text{and} \quad \Delta_{jk} u_k(G, \mathbf{s}) \geq 0.$$
That is, every player with a weak incentive to link with $j$ also has a weak incentive to link with $i$.  Because linking preferences are consistent, the relation $\succeq_{\text{in}}$ is complete---if we do not have $i \succeq_{\text{in}} j$, then there exists $k$ with a weak incentive to link with $j$ but not $i$, and if we do not have $j \succeq_{\text{in}} i$, then there exists $\ell$ with a weak incentive to link with $i$ but not $j$, contradicting consistent link preferences.  It is straightforward to check that this relation is also transitive.  Similarly, the binary relation $\succeq_{\text{out}}$ defined via $i \succeq_{\text{out}} j$ if there is no player $k$ such that
$$\Delta_{ik} u_i(G, \mathbf{s}) < 0 \quad \text{and} \quad \Delta_{jk} u_k(G, \mathbf{s}) \geq 0$$
is complete and transitive.

Suppose there exists $i^*$ maximal under $\succeq_{\text{in}}$, and $i_*$ minimal under $\succeq_{\text{in}}$, such that $i^* \succ_{\text{out}} i_*$.  Fix any players $i$ and $j$.  If $i_*$ wants to link with $j$, then $i^*$ wants to link with $j$ because $i^* \succeq_{\text{out}} i_*$, and $i$ also wants to link with $j$ by alignment.  Similarly, if $i^*$ does not want to link with $j$, then $i_*$ and $i$ do not want to link with $j$.  This tells us that $i^*$ is maximal under $\succeq_{\text{out}}$ and $i_*$ is minimal under $\succeq_{\text{out}}$.

We now show that $i \succeq_{\text{out}} j$ whenever $i \succeq_{\text{in}} j$.  If not, there exists $k$ such that $j$ wants to link with $k$ but $i$ does not.  Since $i^*$ is maximal under $\succeq_{\text{out}}$, we know $i^*$ wants to link with $k$.  We now have
$$i^* \succeq_{\text{in}} i \succeq_{\text{in}} j,$$
and both $i^*$ and $j$ want to link with $k$.  Alignment now implies that $i$ wants to link with $k$, a contradiction.  A similar argument shows that if $i^* \preceq_{\text{out}} i_*$ for any $i^*$ and $i_*$ maximal and minimal respectively under $\succeq_{\text{in}}$, then $i \preceq_{\text{out}} j$ whenever $i \succeq_{\text{in}} j$. \hfill \qedsymbol

\subsection{Supporting results and proofs for  \cref{sec:carrell}} \label{sec:carrell_support}

We begin with a  lemma that characterizes equilibrium actions for players connected in a clique.

Throughout this subsection, we write $\overline{b}_{N'} = \frac{1}{n} \sum_{i \in N'} b_i$ for the average private incentive in a subset $N' \subseteq N$. We write $\overline{b}$ for $\overline{b}_N$, the average incentive in the whole population. Also, recall a clique with nodes $C$ is a graph such that players $i$ and $j$ are linked for every $i,j \in C$.

\begin{lemma}\label{lem:cliqueaction}
Suppose $\alpha \leq 1$, and the stable network $G$ is a clique with nodes $C$.  Each player $i \in C$ has a unique equilibrium action
\begin{equation}\label{eq:cliqueaction}
    s_i = \frac{1}{\alpha + |C|} \left(b_i + \frac{\alpha |C| \overline{b}_C}{\alpha + (1-\alpha)|C|}\right).
\end{equation}
Consequently, the payoff to player $i$ in clique $C$ is
\begin{equation}\label{eq:welfare}
u_i = \frac{1}{2}\frac{|C|}{(\alpha + |C|)^2}\left(b_i + \frac{\alpha |C| \overline{b}_C}{\alpha + (1-\alpha)|C|}\right)^2.
\end{equation}
\end{lemma}

\medskip

One consequence of \cref{lem:cliqueaction} is that welfare is higher in larger cliques only if spillovers are sufficiently strong.  For any $\alpha < 1$, the denominator in \eqref{eq:welfare} contains a higher power of $|C|$ than the numerator.  Holding private incentives fixed, this means that utility declines if the clique becomes large enough. In the case with $\alpha = 1$, equation \eqref{eq:welfare} gives
$$u_i = \frac{|C|}{2}\frac{\left(b_i + |C| \overline{b}_C\right)^2}{(1 + |C|)^2}.$$
If $|C|$ gets larger without significantly affecting the average private incentive $\overline{b}_C$, this payoff increases: the first factor is proportional to $|C|$, and the second factor approaches $\overline{b}_C^2$.

\paragraph{Proof of \cref{lem:cliqueaction}} 

From the first-order condition, we have
$$s_i = \frac{1}{|C|}\left(b_i - \alpha s_i + \alpha \sum_{j \in C} s_j\right) \quad \implies \quad b_i = (|C| + \alpha)s_i - \alpha \sum_{j \in C} s_j.$$
Summing over all $i \in C$, we get
$$|C| \overline{b}_C = (|C| + \alpha - \alpha |C|)\sum_{j \in C} s_j \quad \implies \quad \sum_{j \in C} s_j = \frac{|C| \overline{b}_C}{\alpha + (1-\alpha)|C|}.$$
Substituting into the first-order condition and solving yields the result.

\paragraph{Proof of \cref{prop:empty}}

Recall player $i$ finds it strictly beneficial to link with $j$ if and only if $\frac{s_i}{s_j} < 2\alpha$---the first claim follows.  It is weakly beneficial only if $\frac{s_i}{s_j} \leq 2\alpha$, and this can never hold for both $i$ and $j$ if $\alpha < \frac{1}{2}$.  If $\alpha = \frac{1}{2}$, it holds for both if and only if $s_i = s_j$---if $b_i = b_j$ for some pair $ij$, then $i$ and $J$ can optimally form a clique, so the empty graph is not uniquely stable.  Going forward, we assume the $b_i$ are distinct.

Suppose we $G$ is not empty, and let $S$ be the largest connected component.  Write $\overline{s}$ and $\underline{s}$ for the maximal and minimal equilibrium actions among players in $S$.  If the outcome is stable, we need $\frac{s_i}{s_j} \leq 2 \alpha$ whenever $i$ and $j$ are linked, implying $\frac{\overline{s}}{\underline{s}} \leq (2\alpha)^n$.  Using the first-order condition, we have
$$s_i = \frac{1}{d_i + 1}\left(b_i + \alpha \sum_{j \in G_i} s_j\right),$$
so for $i \in S$, we must have
$$\frac{b_i + \alpha d_i \underline{s}}{d_i+1} \leq s_i \leq \frac{b_i + \alpha d_i \overline{s}}{d_i+1} \leq \frac{b_i + \alpha d_i (2\alpha)^n \underline{s}}{d_i+1}.$$

Within a connected component, there exist two players $i$ and $j$ with the same degree---if there are $m$ players in a component, then there are $m-1$ possible degree values, and the pigeonhole principle tells us two players have the same degree.  Suppose $i,j \in S$ have $d_i = d_j = d$, and without loss of generality assume $b_i < b_j$.  We then have
$$(d+1)s_i \leq b_i + \alpha d \overline{s} \leq b_i + \alpha d (2\alpha)^n \underline{s}, \quad \text{and} \quad (d+1)s_j \geq b_j + \alpha d \underline{s}.$$
Choosing $\alpha > \frac{1}{2}$, but sufficiently close to $\frac{1}{2}$, we can conclude that $(2\alpha)^n s_i <  s_j$, which implies that $i$ and $j$ cannot connected---there exists a link that a player would like to sever.  Hence, for $\alpha$ sufficiently close to $\frac{1}{2}$, the only pairwise stable outcome involves an empty graph. \hfill \qedsymbol

\paragraph{Proof of \cref{prop:complete}}

To ensure that the complete graph is part of a pairwise stable outcome, we only need to check that player $n$, who takes the highest action in the complete graph equilibrium, does not want to delete her link to player $1$, who takes the lowest action.  Using \cref{lem:cliqueaction}, this means
\begin{equation}\label{eq:complete1}
    2\alpha \geq \frac{s_n}{s_1} = \frac{b_n(\alpha + (1-\alpha)n) + \alpha n \overline{b}}{b_1(\alpha + (1-\alpha)n) + \alpha n \overline{b}},
\end{equation}
which is equivalent to
\begin{equation}\label{eq:complete2}
(b_n - 2\alpha b_1)\left(\alpha + (1 - \alpha)n\right) \leq (2\alpha - 1)\alpha n \overline{b}.
\end{equation}
This inequality is hardest to satisfy if we minimize the right hand side.  To do this, we take $b_j = b_1$ for each $j \neq n$.  Hence, the complete graph is part of a pairwise stable outcome if
$$(b_n - 2\alpha b_1) \left(\alpha + (1-\alpha)n\right) \leq \alpha (2\alpha - 1) \left(b_n + (n-1)b_1\right),$$
which is equivalent to
$$b_n(1-\alpha)\left(2 + \frac{n}{\alpha}\right) \leq b_1\left(1 + n\right) \quad \implies \quad \frac{b_n}{b_1} \leq \frac{\alpha(1 + n)}{(1-\alpha)(2\alpha + n)}.$$

For the second claim, we need to show that no other graph can occur in a pairwise stable outcome.  Recall the first-order condition 
$$s_i = \frac{1}{1 + d_i}\left(b_i + \alpha \sum_{j \in G_i} s_j\right).$$
Notice that the best response to any $\mathbf{s}$ in which $s_i \leq b_n$ for each $i$ is a vector $\mathbf{s}'$ in which $s'_i \leq b_n$ for each $i$, so we necessarily have $s_i \leq b_n$ in the unique Nash equilibrium on any fixed graph.  Conversely, let $\underline{s}$ denote the lowest equilibrium action across all graphs $G$ that are not complete.  We necessarily have
$$\underline{s} \geq \frac{1}{d+1}\left(b_1 + \alpha d \underline{s}\right) \quad \implies \quad \underline{s} \geq \frac{b_1}{\alpha + (1-\alpha)(d-1)}$$
for each $d = 0,1, 2, \ldots, n-2$.  The weakest bound occurs for $d = n-2$---we need not consider the case in which the player taking the lowest action has degree $n-1$ because \cref{cor:structure} then implies the graph is complete.  Consequently, the ratio of any two players' actions across all graphs that are not complete is bounded above by
$$\frac{b_n(\alpha + (1-\alpha)(n-1))}{b_1} < 2 \alpha,$$
implying that all players have a strict incentive to link.

For the last claim, we return to equation \eqref{eq:complete2} and \emph{maximize} the right hand size by taking $b_2 = b_3 = \cdots = b_n$.  We find that the complete graph cannot be pairwise stable if
$$(b_n - 2\alpha b_1)\left(\alpha + (1 - \alpha) n\right) > \alpha (2\alpha - 1) \left((n-1)b_n + b_1\right),$$
which is equivalent to the stated condition. \hfill \qedsymbol

}

\end{document}

%% file: illustrations/squadron1.tex
\tikzset{every picture/.style={line width=0.75pt}} %

\begin{tikzpicture}[x=0.75pt,y=0.75pt,yscale=-1,xscale=1]

\draw [color={rgb, 255:red, 42; green, 175; blue, 236 }  ,draw opacity=1 ][line width=3.75]    (328.41,104.55) -- (514.77,239.95) ;
\draw [color={rgb, 255:red, 42; green, 175; blue, 236 }  ,draw opacity=1 ][line width=3.75]    (514.77,239.95) -- (443.59,459.02) ;
\draw [color={rgb, 255:red, 42; green, 175; blue, 236 }  ,draw opacity=1 ][line width=3.75]    (142.06,239.95) -- (328.41,104.55) ;
\draw [color={rgb, 255:red, 42; green, 175; blue, 236 }  ,draw opacity=1 ][line width=3.75]    (142.06,239.95) -- (213.24,459.02) ;
\draw [color={rgb, 255:red, 42; green, 175; blue, 236 }  ,draw opacity=1 ][line width=3.75]    (213.24,459.02) -- (443.59,459.02) ;
\draw [color={rgb, 255:red, 42; green, 175; blue, 236 }  ,draw opacity=1 ][line width=3.75]    (328.41,104.55) -- (443.59,459.02) ;
\draw [color={rgb, 255:red, 42; green, 175; blue, 236 }  ,draw opacity=1 ][line width=3.75]    (142.06,239.95) -- (443.59,459.02) ;
\draw [color={rgb, 255:red, 42; green, 175; blue, 236 }  ,draw opacity=1 ][line width=3.75]    (514.77,239.95) -- (213.24,459.02) ;
\draw [color={rgb, 255:red, 42; green, 175; blue, 236 }  ,draw opacity=1 ][line width=3.75]    (142.06,239.95) -- (514.77,239.95) ;
\draw [color={rgb, 255:red, 42; green, 175; blue, 236 }  ,draw opacity=1 ][line width=3.75]    (328.41,104.55) -- (213.24,459.02) ;
\draw  [color={rgb, 255:red, 42; green, 175; blue, 236 }  ,draw opacity=1 ][fill={rgb, 255:red, 255; green, 255; blue, 255 }  ,fill opacity=1 ][line width=3.75]  (292.3,104.55) .. controls (292.3,84.61) and (308.47,68.44) .. (328.41,68.44) .. controls (348.36,68.44) and (364.53,84.61) .. (364.53,104.55) .. controls (364.53,124.5) and (348.36,140.67) .. (328.41,140.67) .. controls (308.47,140.67) and (292.3,124.5) .. (292.3,104.55) -- cycle ;
\draw  [color={rgb, 255:red, 42; green, 175; blue, 236 }  ,draw opacity=1 ][fill={rgb, 255:red, 255; green, 255; blue, 255 }  ,fill opacity=1 ][line width=3.75]  (478.65,239.95) .. controls (478.65,220) and (494.82,203.83) .. (514.77,203.83) .. controls (534.71,203.83) and (550.88,220) .. (550.88,239.95) .. controls (550.88,259.89) and (534.71,276.06) .. (514.77,276.06) .. controls (494.82,276.06) and (478.65,259.89) .. (478.65,239.95) -- cycle ;
\draw  [color={rgb, 255:red, 42; green, 175; blue, 236 }  ,draw opacity=1 ][fill={rgb, 255:red, 255; green, 255; blue, 255 }  ,fill opacity=1 ][line width=3.75]  (407.47,459.02) .. controls (407.47,439.07) and (423.64,422.91) .. (443.59,422.91) .. controls (463.53,422.91) and (479.7,439.07) .. (479.7,459.02) .. controls (479.7,478.97) and (463.53,495.13) .. (443.59,495.13) .. controls (423.64,495.13) and (407.47,478.97) .. (407.47,459.02) -- cycle ;
\draw  [color={rgb, 255:red, 42; green, 175; blue, 236 }  ,draw opacity=1 ][fill={rgb, 255:red, 255; green, 255; blue, 255 }  ,fill opacity=1 ][line width=3.75]  (177.12,459.02) .. controls (177.12,439.07) and (193.29,422.91) .. (213.24,422.91) .. controls (233.19,422.91) and (249.35,439.07) .. (249.35,459.02) .. controls (249.35,478.97) and (233.19,495.13) .. (213.24,495.13) .. controls (193.29,495.13) and (177.12,478.97) .. (177.12,459.02) -- cycle ;
\draw  [color={rgb, 255:red, 42; green, 175; blue, 236 }  ,draw opacity=1 ][fill={rgb, 255:red, 255; green, 255; blue, 255 }  ,fill opacity=1 ][line width=3.75]  (105.94,239.95) .. controls (105.94,220) and (122.11,203.83) .. (142.06,203.83) .. controls (162,203.83) and (178.17,220) .. (178.17,239.95) .. controls (178.17,259.89) and (162,276.06) .. (142.06,276.06) .. controls (122.11,276.06) and (105.94,259.89) .. (105.94,239.95) -- cycle ;
\draw  [color={rgb, 255:red, 42; green, 175; blue, 236 }  ,draw opacity=1 ][fill={rgb, 255:red, 255; green, 255; blue, 255 }  ,fill opacity=1 ][line width=3.75]  (478.3,240.55) .. controls (478.3,220.61) and (494.47,204.44) .. (514.41,204.44) .. controls (534.36,204.44) and (550.53,220.61) .. (550.53,240.55) .. controls (550.53,260.5) and (534.36,276.67) .. (514.41,276.67) .. controls (494.47,276.67) and (478.3,260.5) .. (478.3,240.55) -- cycle ;
\draw  [color={rgb, 255:red, 42; green, 175; blue, 236 }  ,draw opacity=1 ][fill={rgb, 255:red, 255; green, 255; blue, 255 }  ,fill opacity=1 ][line width=3.75]  (407.3,459.55) .. controls (407.3,439.61) and (423.47,423.44) .. (443.41,423.44) .. controls (463.36,423.44) and (479.53,439.61) .. (479.53,459.55) .. controls (479.53,479.5) and (463.36,495.67) .. (443.41,495.67) .. controls (423.47,495.67) and (407.3,479.5) .. (407.3,459.55) -- cycle ;
\draw  [color={rgb, 255:red, 42; green, 175; blue, 236 }  ,draw opacity=1 ][fill={rgb, 255:red, 255; green, 255; blue, 255 }  ,fill opacity=1 ][line width=3.75]  (177.3,458.55) .. controls (177.3,438.61) and (193.47,422.44) .. (213.41,422.44) .. controls (233.36,422.44) and (249.53,438.61) .. (249.53,458.55) .. controls (249.53,478.5) and (233.36,494.67) .. (213.41,494.67) .. controls (193.47,494.67) and (177.3,478.5) .. (177.3,458.55) -- cycle ;
\draw  [color={rgb, 255:red, 42; green, 175; blue, 236 }  ,draw opacity=1 ][fill={rgb, 255:red, 255; green, 255; blue, 255 }  ,fill opacity=1 ][line width=3.75]  (105.3,239.55) .. controls (105.3,219.61) and (121.47,203.44) .. (141.41,203.44) .. controls (161.36,203.44) and (177.53,219.61) .. (177.53,239.55) .. controls (177.53,259.5) and (161.36,275.67) .. (141.41,275.67) .. controls (121.47,275.67) and (105.3,259.5) .. (105.3,239.55) -- cycle ;

\draw (318.22,94.06) node [anchor=north west][inner sep=0.75pt]  [font=\huge,color={rgb, 255:red, 0; green, 0; blue, 0 }  ,opacity=1 ] [align=left] {$\displaystyle 9$};
\draw (314.22,10.06) node [anchor=north west][inner sep=0.75pt]  [font=\Large] [align=left] {$\displaystyle 6\frac{1}{3}$};
\draw (559.22,217.1) node [anchor=north west][inner sep=0.75pt]  [font=\Large] [align=left] {$\displaystyle 5\frac{5}{6}$};
\draw (66.22,217.1) node [anchor=north west][inner sep=0.75pt]  [font=\Large] [align=left] {$\displaystyle 5\frac{1}{2}$};
\draw (141.22,463.1) node [anchor=north west][inner sep=0.75pt]  [font=\Large] [align=left] {$\displaystyle 5\frac{1}{2}$};
\draw (481.22,463.1) node [anchor=north west][inner sep=0.75pt]  [font=\Large] [align=left] {$\displaystyle 5\frac{5}{6}$};
\draw (505.22,229.06) node [anchor=north west][inner sep=0.75pt]  [font=\huge,color={rgb, 255:red, 0; green, 0; blue, 0 }  ,opacity=1 ] [align=left] {$\displaystyle 6$};
\draw (433.22,448.06) node [anchor=north west][inner sep=0.75pt]  [font=\huge,color={rgb, 255:red, 0; green, 0; blue, 0 }  ,opacity=1 ] [align=left] {$\displaystyle 6$};
\draw (202.22,448.06) node [anchor=north west][inner sep=0.75pt]  [font=\huge,color={rgb, 255:red, 0; green, 0; blue, 0 }  ,opacity=1 ] [align=left] {$\displaystyle 4$};
\draw (130.22,229.06) node [anchor=north west][inner sep=0.75pt]  [font=\huge,color={rgb, 255:red, 0; green, 0; blue, 0 }  ,opacity=1 ] [align=left] {$\displaystyle 4$};

\end{tikzpicture}

%% file: illustrations/squadron2.tex
\tikzset{every picture/.style={line width=0.75pt}} %

\begin{tikzpicture}[x=0.75pt,y=0.75pt,yscale=-1,xscale=1]

\draw [color={rgb, 255:red, 42; green, 175; blue, 236 }  ,draw opacity=1 ][line width=3.75]    (328.41,104.55) -- (514.77,239.95) ;
\draw [color={rgb, 255:red, 42; green, 175; blue, 236 }  ,draw opacity=1 ][line width=3.75]    (514.77,239.95) -- (443.59,459.02) ;
\draw [color={rgb, 255:red, 42; green, 175; blue, 236 }  ,draw opacity=1 ][line width=3.75]    (142.06,239.95) -- (328.41,104.55) ;
\draw [color={rgb, 255:red, 42; green, 175; blue, 236 }  ,draw opacity=1 ][line width=3.75]    (142.06,239.95) -- (213.24,459.02) ;
\draw [color={rgb, 255:red, 42; green, 175; blue, 236 }  ,draw opacity=1 ][line width=3.75]    (213.24,459.02) -- (443.59,459.02) ;
\draw [color={rgb, 255:red, 42; green, 175; blue, 236 }  ,draw opacity=1 ][line width=3.75]    (328.41,104.55) -- (443.59,459.02) ;
\draw [color={rgb, 255:red, 42; green, 175; blue, 236 }  ,draw opacity=1 ][line width=3.75]    (142.06,239.95) -- (443.59,459.02) ;
\draw [color={rgb, 255:red, 42; green, 175; blue, 236 }  ,draw opacity=1 ][line width=3.75]    (514.77,239.95) -- (213.24,459.02) ;
\draw [color={rgb, 255:red, 42; green, 175; blue, 236 }  ,draw opacity=1 ][line width=3.75]    (142.06,239.95) -- (514.77,239.95) ;
\draw [color={rgb, 255:red, 42; green, 175; blue, 236 }  ,draw opacity=1 ][line width=3.75]    (328.41,104.55) -- (213.24,459.02) ;
\draw  [color={rgb, 255:red, 42; green, 175; blue, 236 }  ,draw opacity=1 ][fill={rgb, 255:red, 255; green, 255; blue, 255 }  ,fill opacity=1 ][line width=3.75]  (292.3,104.55) .. controls (292.3,84.61) and (308.47,68.44) .. (328.41,68.44) .. controls (348.36,68.44) and (364.53,84.61) .. (364.53,104.55) .. controls (364.53,124.5) and (348.36,140.67) .. (328.41,140.67) .. controls (308.47,140.67) and (292.3,124.5) .. (292.3,104.55) -- cycle ;
\draw  [color={rgb, 255:red, 42; green, 175; blue, 236 }  ,draw opacity=1 ][fill={rgb, 255:red, 255; green, 255; blue, 255 }  ,fill opacity=1 ][line width=3.75]  (478.65,239.95) .. controls (478.65,220) and (494.82,203.83) .. (514.77,203.83) .. controls (534.71,203.83) and (550.88,220) .. (550.88,239.95) .. controls (550.88,259.89) and (534.71,276.06) .. (514.77,276.06) .. controls (494.82,276.06) and (478.65,259.89) .. (478.65,239.95) -- cycle ;
\draw  [color={rgb, 255:red, 42; green, 175; blue, 236 }  ,draw opacity=1 ][fill={rgb, 255:red, 255; green, 255; blue, 255 }  ,fill opacity=1 ][line width=3.75]  (407.47,459.02) .. controls (407.47,439.07) and (423.64,422.91) .. (443.59,422.91) .. controls (463.53,422.91) and (479.7,439.07) .. (479.7,459.02) .. controls (479.7,478.97) and (463.53,495.13) .. (443.59,495.13) .. controls (423.64,495.13) and (407.47,478.97) .. (407.47,459.02) -- cycle ;
\draw  [color={rgb, 255:red, 42; green, 175; blue, 236 }  ,draw opacity=1 ][fill={rgb, 255:red, 255; green, 255; blue, 255 }  ,fill opacity=1 ][line width=3.75]  (177.12,459.02) .. controls (177.12,439.07) and (193.29,422.91) .. (213.24,422.91) .. controls (233.19,422.91) and (249.35,439.07) .. (249.35,459.02) .. controls (249.35,478.97) and (233.19,495.13) .. (213.24,495.13) .. controls (193.29,495.13) and (177.12,478.97) .. (177.12,459.02) -- cycle ;
\draw  [color={rgb, 255:red, 42; green, 175; blue, 236 }  ,draw opacity=1 ][fill={rgb, 255:red, 255; green, 255; blue, 255 }  ,fill opacity=1 ][line width=3.75]  (105.94,239.95) .. controls (105.94,220) and (122.11,203.83) .. (142.06,203.83) .. controls (162,203.83) and (178.17,220) .. (178.17,239.95) .. controls (178.17,259.89) and (162,276.06) .. (142.06,276.06) .. controls (122.11,276.06) and (105.94,259.89) .. (105.94,239.95) -- cycle ;
\draw  [color={rgb, 255:red, 42; green, 175; blue, 236 }  ,draw opacity=1 ][fill={rgb, 255:red, 255; green, 255; blue, 255 }  ,fill opacity=1 ][line width=3.75]  (478.3,240.55) .. controls (478.3,220.61) and (494.47,204.44) .. (514.41,204.44) .. controls (534.36,204.44) and (550.53,220.61) .. (550.53,240.55) .. controls (550.53,260.5) and (534.36,276.67) .. (514.41,276.67) .. controls (494.47,276.67) and (478.3,260.5) .. (478.3,240.55) -- cycle ;
\draw  [color={rgb, 255:red, 42; green, 175; blue, 236 }  ,draw opacity=1 ][fill={rgb, 255:red, 255; green, 255; blue, 255 }  ,fill opacity=1 ][line width=3.75]  (407.3,459.55) .. controls (407.3,439.61) and (423.47,423.44) .. (443.41,423.44) .. controls (463.36,423.44) and (479.53,439.61) .. (479.53,459.55) .. controls (479.53,479.5) and (463.36,495.67) .. (443.41,495.67) .. controls (423.47,495.67) and (407.3,479.5) .. (407.3,459.55) -- cycle ;
\draw  [color={rgb, 255:red, 42; green, 175; blue, 236 }  ,draw opacity=1 ][fill={rgb, 255:red, 255; green, 255; blue, 255 }  ,fill opacity=1 ][line width=3.75]  (177.3,458.55) .. controls (177.3,438.61) and (193.47,422.44) .. (213.41,422.44) .. controls (233.36,422.44) and (249.53,438.61) .. (249.53,458.55) .. controls (249.53,478.5) and (233.36,494.67) .. (213.41,494.67) .. controls (193.47,494.67) and (177.3,478.5) .. (177.3,458.55) -- cycle ;
\draw  [color={rgb, 255:red, 42; green, 175; blue, 236 }  ,draw opacity=1 ][fill={rgb, 255:red, 255; green, 255; blue, 255 }  ,fill opacity=1 ][line width=3.75]  (105.3,239.55) .. controls (105.3,219.61) and (121.47,203.44) .. (141.41,203.44) .. controls (161.36,203.44) and (177.53,219.61) .. (177.53,239.55) .. controls (177.53,259.5) and (161.36,275.67) .. (141.41,275.67) .. controls (121.47,275.67) and (105.3,259.5) .. (105.3,239.55) -- cycle ;

\draw (318.22,94.06) node [anchor=north west][inner sep=0.75pt]  [font=\huge,color={rgb, 255:red, 0; green, 0; blue, 0 }  ,opacity=1 ] [align=left] {$\displaystyle 9$};
\draw (314.22,10.06) node [anchor=north west][inner sep=0.75pt]  [font=\Large] [align=left] {$\displaystyle 6\frac{5}{6}$};
\draw (559.22,217.1) node [anchor=north west][inner sep=0.75pt]  [font=\Large] [align=left] {$\displaystyle 6\frac{5}{6}$};
\draw (481.22,463.1) node [anchor=north west][inner sep=0.75pt]  [font=\Large] [align=left] {$\displaystyle 6\frac{1}{3}$};
\draw (505.22,229.06) node [anchor=north west][inner sep=0.75pt]  [font=\huge,color={rgb, 255:red, 0; green, 0; blue, 0 }  ,opacity=1 ] [align=left] {$\displaystyle 9$};
\draw (433.22,448.06) node [anchor=north west][inner sep=0.75pt]  [font=\huge,color={rgb, 255:red, 0; green, 0; blue, 0 }  ,opacity=1 ] [align=left] {$\displaystyle 6$};
\draw (202.22,448.06) node [anchor=north west][inner sep=0.75pt]  [font=\huge,color={rgb, 255:red, 0; green, 0; blue, 0 }  ,opacity=1 ] [align=left] {$\displaystyle 4$};
\draw (130.22,229.06) node [anchor=north west][inner sep=0.75pt]  [font=\huge,color={rgb, 255:red, 0; green, 0; blue, 0 }  ,opacity=1 ] [align=left] {$\displaystyle 4$};
\draw (80.22,235.1) node [anchor=north west][inner sep=0.75pt]  [font=\Large] [align=left] {$\displaystyle 6$};
\draw (155.22,481.1) node [anchor=north west][inner sep=0.75pt]  [font=\Large] [align=left] {$\displaystyle 6$};

\end{tikzpicture}

%% file: illustrations/squadron3.tex
\tikzset{every picture/.style={line width=0.75pt}} %

\begin{tikzpicture}[x=0.75pt,y=0.75pt,yscale=-1,xscale=1]

\draw [color={rgb, 255:red, 42; green, 175; blue, 236 }  ,draw opacity=1 ][line width=3.75]    (328.41,104.55) -- (514.77,239.95) ;
\draw [color={rgb, 255:red, 42; green, 175; blue, 236 }  ,draw opacity=1 ][line width=3.75]    (514.77,239.95) -- (443.59,459.02) ;
\draw [color={rgb, 255:red, 42; green, 175; blue, 236 }  ,draw opacity=1 ][line width=3.75]    (142.06,239.95) -- (213.24,459.02) ;
\draw [color={rgb, 255:red, 42; green, 175; blue, 236 }  ,draw opacity=1 ][line width=3.75]    (328.41,104.55) -- (443.59,459.02) ;
\draw  [color={rgb, 255:red, 42; green, 175; blue, 236 }  ,draw opacity=1 ][fill={rgb, 255:red, 255; green, 255; blue, 255 }  ,fill opacity=1 ][line width=3.75]  (292.3,104.55) .. controls (292.3,84.61) and (308.47,68.44) .. (328.41,68.44) .. controls (348.36,68.44) and (364.53,84.61) .. (364.53,104.55) .. controls (364.53,124.5) and (348.36,140.67) .. (328.41,140.67) .. controls (308.47,140.67) and (292.3,124.5) .. (292.3,104.55) -- cycle ;
\draw  [color={rgb, 255:red, 42; green, 175; blue, 236 }  ,draw opacity=1 ][fill={rgb, 255:red, 255; green, 255; blue, 255 }  ,fill opacity=1 ][line width=3.75]  (478.65,239.95) .. controls (478.65,220) and (494.82,203.83) .. (514.77,203.83) .. controls (534.71,203.83) and (550.88,220) .. (550.88,239.95) .. controls (550.88,259.89) and (534.71,276.06) .. (514.77,276.06) .. controls (494.82,276.06) and (478.65,259.89) .. (478.65,239.95) -- cycle ;
\draw  [color={rgb, 255:red, 42; green, 175; blue, 236 }  ,draw opacity=1 ][fill={rgb, 255:red, 255; green, 255; blue, 255 }  ,fill opacity=1 ][line width=3.75]  (407.47,459.02) .. controls (407.47,439.07) and (423.64,422.91) .. (443.59,422.91) .. controls (463.53,422.91) and (479.7,439.07) .. (479.7,459.02) .. controls (479.7,478.97) and (463.53,495.13) .. (443.59,495.13) .. controls (423.64,495.13) and (407.47,478.97) .. (407.47,459.02) -- cycle ;
\draw  [color={rgb, 255:red, 42; green, 175; blue, 236 }  ,draw opacity=1 ][fill={rgb, 255:red, 255; green, 255; blue, 255 }  ,fill opacity=1 ][line width=3.75]  (177.12,459.02) .. controls (177.12,439.07) and (193.29,422.91) .. (213.24,422.91) .. controls (233.19,422.91) and (249.35,439.07) .. (249.35,459.02) .. controls (249.35,478.97) and (233.19,495.13) .. (213.24,495.13) .. controls (193.29,495.13) and (177.12,478.97) .. (177.12,459.02) -- cycle ;
\draw  [color={rgb, 255:red, 42; green, 175; blue, 236 }  ,draw opacity=1 ][fill={rgb, 255:red, 255; green, 255; blue, 255 }  ,fill opacity=1 ][line width=3.75]  (105.94,239.95) .. controls (105.94,220) and (122.11,203.83) .. (142.06,203.83) .. controls (162,203.83) and (178.17,220) .. (178.17,239.95) .. controls (178.17,259.89) and (162,276.06) .. (142.06,276.06) .. controls (122.11,276.06) and (105.94,259.89) .. (105.94,239.95) -- cycle ;
\draw  [color={rgb, 255:red, 42; green, 175; blue, 236 }  ,draw opacity=1 ][fill={rgb, 255:red, 255; green, 255; blue, 255 }  ,fill opacity=1 ][line width=3.75]  (478.3,240.55) .. controls (478.3,220.61) and (494.47,204.44) .. (514.41,204.44) .. controls (534.36,204.44) and (550.53,220.61) .. (550.53,240.55) .. controls (550.53,260.5) and (534.36,276.67) .. (514.41,276.67) .. controls (494.47,276.67) and (478.3,260.5) .. (478.3,240.55) -- cycle ;
\draw  [color={rgb, 255:red, 42; green, 175; blue, 236 }  ,draw opacity=1 ][fill={rgb, 255:red, 255; green, 255; blue, 255 }  ,fill opacity=1 ][line width=3.75]  (407.3,459.55) .. controls (407.3,439.61) and (423.47,423.44) .. (443.41,423.44) .. controls (463.36,423.44) and (479.53,439.61) .. (479.53,459.55) .. controls (479.53,479.5) and (463.36,495.67) .. (443.41,495.67) .. controls (423.47,495.67) and (407.3,479.5) .. (407.3,459.55) -- cycle ;
\draw  [color={rgb, 255:red, 42; green, 175; blue, 236 }  ,draw opacity=1 ][fill={rgb, 255:red, 255; green, 255; blue, 255 }  ,fill opacity=1 ][line width=3.75]  (177.3,458.55) .. controls (177.3,438.61) and (193.47,422.44) .. (213.41,422.44) .. controls (233.36,422.44) and (249.53,438.61) .. (249.53,458.55) .. controls (249.53,478.5) and (233.36,494.67) .. (213.41,494.67) .. controls (193.47,494.67) and (177.3,478.5) .. (177.3,458.55) -- cycle ;
\draw  [color={rgb, 255:red, 42; green, 175; blue, 236 }  ,draw opacity=1 ][fill={rgb, 255:red, 255; green, 255; blue, 255 }  ,fill opacity=1 ][line width=3.75]  (105.3,239.55) .. controls (105.3,219.61) and (121.47,203.44) .. (141.41,203.44) .. controls (161.36,203.44) and (177.53,219.61) .. (177.53,239.55) .. controls (177.53,259.5) and (161.36,275.67) .. (141.41,275.67) .. controls (121.47,275.67) and (105.3,259.5) .. (105.3,239.55) -- cycle ;

\draw (318.22,94.06) node [anchor=north west][inner sep=0.75pt]  [font=\huge,color={rgb, 255:red, 0; green, 0; blue, 0 }  ,opacity=1 ] [align=left] {$\displaystyle 9$};
\draw (505.22,229.06) node [anchor=north west][inner sep=0.75pt]  [font=\huge,color={rgb, 255:red, 0; green, 0; blue, 0 }  ,opacity=1 ] [align=left] {$\displaystyle 9$};
\draw (433.22,448.06) node [anchor=north west][inner sep=0.75pt]  [font=\huge,color={rgb, 255:red, 0; green, 0; blue, 0 }  ,opacity=1 ] [align=left] {$\displaystyle 9$};
\draw (202.22,448.06) node [anchor=north west][inner sep=0.75pt]  [font=\huge,color={rgb, 255:red, 0; green, 0; blue, 0 }  ,opacity=1 ] [align=left] {$\displaystyle 4$};
\draw (130.22,229.06) node [anchor=north west][inner sep=0.75pt]  [font=\huge,color={rgb, 255:red, 0; green, 0; blue, 0 }  ,opacity=1 ] [align=left] {$\displaystyle 4$};
\draw (80.22,235.1) node [anchor=north west][inner sep=0.75pt]  [font=\Large] [align=left] {$\displaystyle 4$};
\draw (155.22,481.1) node [anchor=north west][inner sep=0.75pt]  [font=\Large] [align=left] {$\displaystyle 4$};
\draw (322.22,27.1) node [anchor=north west][inner sep=0.75pt]  [font=\Large] [align=left] {$\displaystyle 9$};
\draw (561.22,235.1) node [anchor=north west][inner sep=0.75pt]  [font=\Large] [align=left] {$\displaystyle 9$};
\draw (482.22,480.1) node [anchor=north west][inner sep=0.75pt]  [font=\Large] [align=left] {$\displaystyle 9$};

\end{tikzpicture}

%% file: main.bbl
\begin{thebibliography}{47}
\newcommand{\enquote}[1]{``#1''}
\expandafter\ifx\csname natexlab\endcsname\relax\def\natexlab#1{#1}\fi

\bibitem[\protect\citeauthoryear{Adler and Adler}{Adler and Adler}{1995}]{adler1995dynamics}
\textsc{Adler, P.~A. and P.~Adler} (1995): \enquote{Dynamics of inclusion and exclusion in preadolescent cliques,} \emph{Social Psychology Quarterly}, 145--162.

\bibitem[\protect\citeauthoryear{Aumann and Myerson}{Aumann and Myerson}{1988}]{aumann1988endogenous}
\textsc{Aumann, R. and R.~Myerson} (1988): \enquote{Endogenous Formation of Links Between Players and Coalitions: An Application of the Shapley Value,} in \emph{The Shapley Value}, ed. by A.~Roth, Cambridge University Press, 175--191.

\bibitem[\protect\citeauthoryear{Baccara and Yariv}{Baccara and Yariv}{2013}]{baccara2013homophily}
\textsc{Baccara, M. and L.~Yariv} (2013): \enquote{Homophily in peer groups,} \emph{American Economic Journal: Microeconomics}, 5, 69--96.

\bibitem[\protect\citeauthoryear{Badev}{Badev}{2021}]{Badev2021}
\textsc{Badev, A.} (2021): \enquote{{Nash equilibria on (un)stable networks},} \emph{Econometrica}, 89, 1179--1206.

\bibitem[\protect\citeauthoryear{Baetz}{Baetz}{2015}]{Baetz2015}
\textsc{Baetz, O.} (2015): \enquote{{Social activity and network formation},} \emph{Theoretical Economics}, 10, 315--340.

\bibitem[\protect\citeauthoryear{Ballester, Calv\'{o}-Armengol, and Zenou}{Ballester et~al.}{2006}]{Ballesteretal2006}
\textsc{Ballester, C., A.~Calv\'{o}-Armengol, and Y.~Zenou} (2006): \enquote{{Who's who in networks. Wanted: the key player},} \emph{Econometrica}, 74, 1403--1417.

\bibitem[\protect\citeauthoryear{Bandyopadhyay and Cabrales}{Bandyopadhyay and Cabrales}{2020}]{bandyopadhyay2020pricing}
\textsc{Bandyopadhyay, S. and A.~Cabrales} (2020): \enquote{Pricing group membership,} CESifo Working Paper.

\bibitem[\protect\citeauthoryear{Belhaj, Bervoets, and Dero{\"\i}an}{Belhaj et~al.}{2016}]{belhaj2016efficient}
\textsc{Belhaj, M., S.~Bervoets, and F.~Dero{\"\i}an} (2016): \enquote{Efficient networks in games with local complementarities,} \emph{Theoretical Economics}, 11, 357--380.

\bibitem[\protect\citeauthoryear{Bloch and Jackson}{Bloch and Jackson}{2006}]{bloch2006definitions}
\textsc{Bloch, F. and M.~O. Jackson} (2006): \enquote{Definitions of equilibrium in network formation games,} \emph{International Journal of Game Theory}, 34, 305--318.

\bibitem[\protect\citeauthoryear{Bolletta}{Bolletta}{2021}]{bolletta2021model}
\textsc{Bolletta, U.} (2021): \enquote{A model of peer effects in school,} \emph{Mathematical Social Sciences}, 114, 1--10.

\bibitem[\protect\citeauthoryear{Bramoull{\'e} and Ghiglino}{Bramoull{\'e} and Ghiglino}{2022}]{bramoulle2022loss}
\textsc{Bramoull{\'e}, Y. and C.~Ghiglino} (2022): \enquote{Loss aversion and conspicuous consumption in networks,} CEPR Discussion Paper No. DP17181.

\bibitem[\protect\citeauthoryear{Bramoull\'{e} and Kranton}{Bramoull\'{e} and Kranton}{2007}]{BramoulleKranton2007}
\textsc{Bramoull\'{e}, Y. and R.~Kranton} (2007): \enquote{{Public goods in networks},} \emph{Journal of Economic Theory}, 135, 478--494.

\bibitem[\protect\citeauthoryear{Bramoull\'{e}, Kranton, and {D'Amours}}{Bramoull\'{e} et~al.}{2014}]{Bramoulleetal2014}
\textsc{Bramoull\'{e}, Y., R.~Kranton, and M.~{D'Amours}} (2014): \enquote{{Strategic interaction and networks},} \emph{American Economic Review}, 104, 898--930.

\bibitem[\protect\citeauthoryear{Cabrales, {Calv\'{o}-Armengol}, and Zenou}{Cabrales et~al.}{2011}]{Cabralesetal2011}
\textsc{Cabrales, A., A.~{Calv\'{o}-Armengol}, and Y.~Zenou} (2011): \enquote{{Social interactions and spillovers},} \emph{Games and Economic Behavior}, 72, 339--360.

\bibitem[\protect\citeauthoryear{{Calv\'{o}-Armengol} and \.{I}lk{\i}l{\i}\c{c}}{{Calv\'{o}-Armengol} and \.{I}lk{\i}l{\i}\c{c}}{2009}]{CalvoArmengolIlkilic2009}
\textsc{{Calv\'{o}-Armengol}, A. and R.~\.{I}lk{\i}l{\i}\c{c}} (2009): \enquote{{Pairwise-Stability and Nash Equilibrium in Network Formation},} \emph{International Journal of Game Theory}, 38, 51--79.

\bibitem[\protect\citeauthoryear{Carrell, Sacerdote, and West}{Carrell et~al.}{2013}]{Carrelletal2013}
\textsc{Carrell, S., B.~Sacerdote, and J.~West} (2013): \enquote{{From natural variation to optimal policy? The importance of endogenous peer group formation},} \emph{Econometrica}, 81, 855--882.

\bibitem[\protect\citeauthoryear{Cerreia-Vioglio, Corrao, and Lanzani}{Cerreia-Vioglio et~al.}{2023}]{CerreiaVoglio2023Nonlinear}
\textsc{Cerreia-Vioglio, S., R.~Corrao, and G.~Lanzani} (2023): \enquote{Nonlinear Fixed Points and Stationarity: Economic Applications,} Working Paper, MIT, available at \url{https://economics.mit.edu/sites/default/files/inline-files/POSTING.pdf}.

\bibitem[\protect\citeauthoryear{Chade and Eeckhout}{Chade and Eeckhout}{2018}]{chade2018matching}
\textsc{Chade, H. and J.~Eeckhout} (2018): \enquote{Matching information,} \emph{Theoretical Economics}, 13, 377--414.

\bibitem[\protect\citeauthoryear{Chakrabarti and Gilles}{Chakrabarti and Gilles}{2007}]{ChakrabartiGilles2007}
\textsc{Chakrabarti, S. and R.~Gilles} (2007): \enquote{{Network potentials},} \emph{Review of Economic Design}, 11, 13--52.

\bibitem[\protect\citeauthoryear{Chandrasekhar and Jackson}{Chandrasekhar and Jackson}{2021}]{chandrasekhar2021network}
\textsc{Chandrasekhar, A.~G. and M.~O. Jackson} (2021): \enquote{A network formation model based on subgraphs,} \emph{Available at SSRN 2660381}.

\bibitem[\protect\citeauthoryear{Davis and Leinhardt}{Davis and Leinhardt}{1971}]{davis1967structure}
\textsc{Davis, J.~A. and S.~Leinhardt} (1971): \enquote{The structure of positive interpersonal relations in small groups,} in \emph{Sociological Theories in Progress}, ed. by J.~Berger, M.~Zelditch, and B.~Anderson, Houghton-Mifflin, vol.~2.

\bibitem[\protect\citeauthoryear{Galeotti and Goyal}{Galeotti and Goyal}{2010}]{GaleottiGoyal2010}
\textsc{Galeotti, A. and S.~Goyal} (2010): \enquote{{The law of the few},} \emph{American Economic Review}, 100, 1468--1492.

\bibitem[\protect\citeauthoryear{Gest, Davidson, Rulison, Moody, and Welsh}{Gest et~al.}{2007}]{gest2007features}
\textsc{Gest, S.~D., A.~J. Davidson, K.~L. Rulison, J.~Moody, and J.~A. Welsh} (2007): \enquote{Features of groups and status hierarchies in girls' and boys' early adolescent peer networks,} \emph{New Directions for Child and Adolescent Development}, 2007, 43--60.

\bibitem[\protect\citeauthoryear{Ghiglino and Goyal}{Ghiglino and Goyal}{2010}]{ghiglino2010keeping}
\textsc{Ghiglino, C. and S.~Goyal} (2010): \enquote{Keeping up with the neighbors: social interaction in a market economy,} \emph{Journal of the European Economic Association}, 8, 90--119.

\bibitem[\protect\citeauthoryear{Golub, Hsieh, and Sadler}{Golub et~al.}{2023}]{golub2023difficulty}
\textsc{Golub, B., Y.-C. Hsieh, and E.~Sadler} (2023): \enquote{On the difficulty of characterizing network formation with endogenous behavior,} \emph{Mathematical Social Sciences}.

\bibitem[\protect\citeauthoryear{Goyal and Joshi}{Goyal and Joshi}{2006}]{GoyalJoshi2006}
\textsc{Goyal, S. and S.~Joshi} (2006): \enquote{{Unequal connections},} \emph{International Journal of Game Theory}, 34, 319--349.

\bibitem[\protect\citeauthoryear{Hellmann}{Hellmann}{2013}]{Hellmann2013}
\textsc{Hellmann, T.} (2013): \enquote{{On the existence and uniqueness of pairwise stable networks},} \emph{International Journal of Game Theory}, 42, 211--237.

\bibitem[\protect\citeauthoryear{Hellmann}{Hellmann}{2020}]{Hellmann2020}
---\hspace{-.1pt}---\hspace{-.1pt}--- (2020): \enquote{{Pairwise stable networks in homogeneous societies with weak link externalities},} Forthcoming in European Journal of Operational Research.

\bibitem[\protect\citeauthoryear{Herskovic and Ramos}{Herskovic and Ramos}{2020}]{HerskovicRamos2020}
\textsc{Herskovic, B. and J.~Ramos} (2020): \enquote{{Acquiring information through peers},} \emph{American Economic Review}, 110, 2128--2152.

\bibitem[\protect\citeauthoryear{Hiller}{Hiller}{2017}]{Hiller2017}
\textsc{Hiller, T.} (2017): \enquote{{Peer effects in endogenous networks},} \emph{Games and Economic Behavior}, 105, 349--367.

\bibitem[\protect\citeauthoryear{Holme and Newman}{Holme and Newman}{2006}]{holme2006nonequilibrium}
\textsc{Holme, P. and M.~E. Newman} (2006): \enquote{Nonequilibrium phase transition in the coevolution of networks and opinions,} \emph{Physical Review E}, 74, 056108.

\bibitem[\protect\citeauthoryear{Homans}{Homans}{1950}]{homans1950human}
\textsc{Homans, G.~C.} (1950): \emph{The Human Group}, New York: Harcourt.

\bibitem[\protect\citeauthoryear{\.{I}lk{\i}l{\i}\c{c} and \.{I}kizler}{\.{I}lk{\i}l{\i}\c{c} and \.{I}kizler}{2019}]{IlkilicIkizler2019}
\textsc{\.{I}lk{\i}l{\i}\c{c}, R. and H.~\.{I}kizler} (2019): \enquote{{Equilibrium Refinements for the Network Formation Game},} \emph{Review of Economic Design}, 23, 13--25.

\bibitem[\protect\citeauthoryear{Immorlica, Kranton, Manea, and Stoddard}{Immorlica et~al.}{2017}]{Immorlicaetal2017}
\textsc{Immorlica, N., R.~Kranton, M.~Manea, and G.~Stoddard} (2017): \enquote{{Social status in networks},} \emph{American Economic Journal: Microeconomics}, 9, 1--30.

\bibitem[\protect\citeauthoryear{Jackson}{Jackson}{2008}]{jackson2008}
\textsc{Jackson, M.} (2008): \emph{{Social and Economic Networks}}, Princeton University Press.

\bibitem[\protect\citeauthoryear{Jackson and Watts}{Jackson and Watts}{2001}]{JacksonWatts2001}
\textsc{Jackson, M. and A.~Watts} (2001): \enquote{{The existence of pairwise stable networks},} \emph{Seoul Journal of Economics}, 14, 299--321.

\bibitem[\protect\citeauthoryear{Jackson and Wolinsky}{Jackson and Wolinsky}{1996}]{JacksonWolinsky1996}
\textsc{Jackson, M. and A.~Wolinsky} (1996): \enquote{{A strategic model of social and economic networks},} \emph{Journal of Economic Theory}, 71, 44--74.

\bibitem[\protect\citeauthoryear{Jackson}{Jackson}{2019}]{jackson2019friendship}
\textsc{Jackson, M.~O.} (2019): \enquote{The friendship paradox and systematic biases in perceptions and social norms,} \emph{Journal of Political Economy}, 127, 777--818.

\bibitem[\protect\citeauthoryear{Jackson and Watts}{Jackson and Watts}{2002}]{jackson2002formation}
\textsc{Jackson, M.~O. and A.~Watts} (2002): \enquote{On the formation of interaction networks in social coordination games,} \emph{Games and Economic Behavior}, 41, 265--291.

\bibitem[\protect\citeauthoryear{Johnson and Gilles}{Johnson and Gilles}{2000}]{JohnsonGilles2000}
\textsc{Johnson, C. and R.~Gilles} (2000): \enquote{{Spatial social networks},} \emph{Review of Economic Design}, 5, 273--299.

\bibitem[\protect\citeauthoryear{Joshi and Mahmud}{Joshi and Mahmud}{2016}]{JoshiMahmud2016}
\textsc{Joshi, S. and A.~Mahmud} (2016): \enquote{{Network formation under multiple sources of externalities},} \emph{Journal of Public Economic Theory}, 18, 148--167.

\bibitem[\protect\citeauthoryear{Kivinen}{Kivinen}{2017}]{kivinen2017polarization}
\textsc{Kivinen, S.} (2017): \enquote{Polarization in strategic networks,} \emph{Economics Letters}, 154, 81--83.

\bibitem[\protect\citeauthoryear{K\"{o}nig, Tessone, and Zenou}{K\"{o}nig et~al.}{2014}]{Konigetal2014}
\textsc{K\"{o}nig, M., C.~Tessone, and Y.~Zenou} (2014): \enquote{{Nestedness in networks: a theoretical model and some applications},} \emph{Theoretical Economics}, 9, 695--752.

\bibitem[\protect\citeauthoryear{Kranton and Minehart}{Kranton and Minehart}{2001}]{kranton2001theory}
\textsc{Kranton, R.~E. and D.~F. Minehart} (2001): \enquote{A theory of buyer-seller networks,} \emph{American economic review}, 91, 485--508.

\bibitem[\protect\citeauthoryear{Polanski and Vega-Redondo}{Polanski and Vega-Redondo}{2023}]{polanski2023homophily}
\textsc{Polanski, A. and F.~Vega-Redondo} (2023): \enquote{Homophily and influence,} \emph{Journal of Economic Theory}, 207, 105576.

\bibitem[\protect\citeauthoryear{Sadler}{Sadler}{2022}]{Sadler2020a}
\textsc{Sadler, E.} (2022): \enquote{{Ordinal centrality},} \emph{Journal of Political Economy}, 130, 926--955.

\bibitem[\protect\citeauthoryear{Sadler}{Sadler}{2023}]{Sadler2023}
---\hspace{-.1pt}---\hspace{-.1pt}--- (2023): \enquote{{Making a swap: network formation with increasing marginal costs},} Working Paper, available at SSRN 4294169.

\end{thebibliography}
